\documentclass[aps,prxquantum,reprint,superscriptaddress]{revtex4-2}

\usepackage{amsmath}
\usepackage{amsthm}
\usepackage{amsfonts}
\usepackage{amssymb}
\usepackage{color}
\usepackage{dsfont}
\usepackage{graphicx}
\usepackage{latexsym}
\usepackage{mathrsfs}
\usepackage{mathtools}
\usepackage{mleftright}
\usepackage{physics2}
\usepackage{tikz}

\usetikzlibrary{arrows.meta}

\theoremstyle{plain}
\newtheorem{thm}{Theorem}
\newtheorem{cor}[thm]{Corollary}
\newtheorem{lem}[thm]{Lemma}
\newtheorem{prop}[thm]{Proposition}

\theoremstyle{definition}
\newtheorem{dfn}{Definition}

\theoremstyle{remark}
\newtheorem{rmk}{Remark}

\mleftright

\usephysicsmodule{ab,ab.legacy,braket}
\newcommand{\p}{\pab}
\newcommand{\B}{\Bab}
\newcommand{\V}{\Vab}
\renewcommand{\hat}{\widehat}
\renewcommand{\mid}{\:\middle|\:}
\renewcommand{\tilde}{\widetilde}

\begin{document}

\title{Information Thermodynamics in Generalized Probabilistic Theories}

\author{Koki Ono}
\email{ono-koki667@g.ecc.u-tokyo.ac.jp}
\affiliation{Department of Basic Science, The University of Tokyo, 3-8-1 Komaba, Meguro-ku, Tokyo 153-8902, Japan}

\author{Shun Umekawa}
\email{umeshun2003@g.ecc.u-tokyo.ac.jp}
\affiliation{Department of Physics, The University of Tokyo, 5-1-5 Kashiwanoha, Kashiwa, Chiba 277-8574, Japan}

\author{Hiroyasu Tajima}
\email{hiroyasu.tajima@inf.kyushu-u.ac.jp\\K. O. and S. U. contributed equally to this work.}
\affiliation{Department of Informatics, Faculty of Information Science and Electrical Engineering, Kyushu University, 744 Motooka, Nishi-ku, Fukuoka, 819-0395, Japan}
\affiliation{JST, FOREST, 4-1-8 Honcho, Kawaguchi, Saitama, 332-0012, Japan.}

\begin{abstract}
Generalized probabilistic theories (GPTs) provide a unified framework for describing all probabilistic physical theories, encompassing not only classical and quantum theories, but also hypothetical theories beyond quantum mechanics, by abstracting the operational structure of experiments and observations.
However, most GPTs are highly unrealistic and far removed from known physical theories, making it important to constrain them by imposing physically reasonable principles.
One of the most important such principles is consistency with thermodynamics, which has been extensively studied through toy models involving semipermeable membranes (SPMs) implementing measurements.
On the other hand, information thermodynamics, which plays a central role in understanding the relationship between measurement and thermodynamics in classical and quantum theory, has remained largely undeveloped in the context of GPTs.
In this work, we construct information thermodynamics in GPTs and provide a unified framework for analyzing the relationship between measurement, feedback, information erasure, and the second law of thermodynamics.
We also formulate a general framework for SPM models and analyze the thermodynamic cost of the measurement processes implemented by SPMs.
As a result, we show that even in GPTs, no work can be extracted in contradiction with the second law of thermodynamics as long as the measurement processes are consistent with entropy nondecrease.
We further derive sufficient conditions for measurement processes to satisfy entropy nondecrease with respect to several entropy definitions proposed in GPTs.
Moreover, by considering measurement processes that violate these conditions, we construct for the first time explicit GPT systems realizing isothermal SPM cycles from which positive work can be extracted.
These examples demonstrate that, if one assumes the existence of SPMs that implement a certain class of measurements without thermodynamic cost, violations of the second law can arise from the lack of fundamental entropy properties or from discrepancies between different entropy definitions.
These results provide a unified and model-independent foundation for understanding the relationship between thermodynamics and measurement in GPTs.
\end{abstract}

\maketitle

\section{Introduction}

Understanding the fundamental laws that govern the physical world is one of the central goals of physics.
To date, all experimentally verified physical phenomena are believed to be described by quantum theory.
Nevertheless, attempts to reconstruct the axioms of quantum theory from physically motivated principles \cite{Araki1980,Hardy2001,Masanes2011,Chiribella2011,Muller2012-1,Barnum2023} have not yet achieved complete success, and the possibility of physical theories beyond quantum theory therefore remains open.
One notable indication of this possibility is the longstanding incompatibility between quantum theory and general relativity, the latter being the current theory of classical gravity.
Despite decades of research, a consistent unification of these two frameworks has not yet been achieved, suggesting that gravity may not be fully captured within the standard framework of quantum theory \cite{Muller2012-2,Plavala2025}.

The possibility of physical theories beyond quantum theory motivates the study of generalized probabilistic theories (GPTs) \cite{Ludwig1964,Davies1970,Gudder1973,Barrett2007,Janotta2014,Plavala2023}, which provide a framework based only on the minimal operational requirements necessary for probabilistic predictions of measurement outcomes.
By abstracting the operational structure of experiments and observations, GPTs encompass both classical and quantum theories, as well as, in principle, all theories capable of probabilistically describing physical events \cite{Popescu1994,Barrett2007,Janotta2011,Jordan1934,Barnum2020,Yoshida2020,Plavala2022,Plavala2023-2,Jiang2024}.
Even if physical phenomena beyond quantum theory exist, they are expected to obey the operational structure of GPTs.

However, GPTs are extremely general and therefore possess limited predictive power.
Most GPTs are highly unrealistic and differ significantly from known physical theories such as classical and quantum theories.
It is therefore important to impose physically motivated constraints in order to narrow down the space of admissible theories.
Such an approach may enable deeper insights into the structure of physically plausible theories beyond quantum theory.

A particularly important guiding principle for restricting physically admissible theories is consistency with thermodynamics.
As Einstein remarked \cite{Einstein1979}, thermodynamics is among the most reliable physical theories, having survived the major conceptual revolutions brought about by relativity and quantum theory without requiring fundamental revision.
Therefore, GPTs that violate the laws of thermodynamics are unlikely to provide realistic descriptions of the physical world.

Motivated by this perspective, research aiming to constrain GPTs by imposing the laws of thermodynamics has mainly been carried out using the semipermeable-membrane (SPM) models \cite{Hanggi2013,Krumm2015,Krumm2017,Minagawa2022}.
In an SPM model, one considers a classical ideal gas whose internal degrees of freedom are described by a GPT system, together with semipermeable membranes (SPMs) that measure the internal states and transmit or reflect particles depending on the measurement outcomes.
The origin of SPM models can be traced back to von Neumann's thought experiment in quantum theory \cite{Neumann1927,Neumann1932,Petz2001}.
By considering the work required to separate a gas according to its internal quantum states using SPMs, he derived the von Neumann entropy as the thermodynamic entropy of quantum systems.
H\"anggi and Wehner \cite{Hanggi2013} investigated the amount of work extractable from a particular isothermal cycle in the SPM model within quantum theory and related the second law of thermodynamics to uncertainty relations.
They further extended their analysis to a certain class of GPTs, thereby paving the way for investigating the relationship between thermodynamics and GPTs.
Krumm \textit{et al.}\ \cite{Barnum2015,Krumm2015,Krumm2017,Takakura2019} were the first to investigate the consistency between thermodynamics and GPTs through an SPM model.
In GPTs, there is no canonical notion of entropy satisfying all the properties of the von Neumann entropy, and the definition of entropy is therefore not unique \cite{Short2010,Barnum2010,Kimura2010,Kimura2016}.
They derived necessary conditions that a ``thermodynamic entropy'' in a GPT must satisfy, and showed that some GPT systems admit no entropy satisfying these conditions.
They also showed that, under rather strong assumptions on the GPT system, such entropies do exist, satisfy desirable properties, and are consistent with the second law of thermodynamics.
Furthermore, Minagawa \textit{et al.}\ \cite{Minagawa2022} showed that, for GPT systems satisfying certain assumptions, positive work can be extracted from an isothermal cycle in the SPM model, leading to a contradiction with the second law of thermodynamics.

However, the SPM models represent only a highly restricted class of thermodynamic scenarios.
In classical and quantum theory, much more general thermodynamic processes involving measurements and feedback have been extensively studied \cite{Maxwell1872,Szilard1929,Brillouin1951,Gabor1961,Landauer1961,Bennett1982,Maruyama2005,Sagawa2008,Sagawa2009,Sagawa2011,Jacobs2009,Sagawa2010,Toyabe2010,Berut2012,Tajima2013,Funo2013,Parrondo2015,Morikuni2017,Abdelkhalek2018,Proesmans2020,Koshihara2022,Minagawa2025}.
It is therefore natural to expect that, also in the context of GPTs, a more comprehensive understanding of the relationship between thermodynamics and GPTs can be achieved by considering general thermodynamic processes involving measurements, rather than restricting attention to specific models such as the SPM model.

The study of such thermodynamic processes involving measurements originates from the problem of Maxwell's demon \cite{Maxwell1872}, which pointed out that allowing unrestricted measurement-based processes could lead to apparent violations of the second law of thermodynamics.
Szilard \cite{Szilard1929} clarified the essence of this problem through a simple model involving measurement and feedback, thereby suggesting a deep connection between thermodynamics and information.
Motivated by this idea, Brillouin \cite{Brillouin1951} and Gabor \cite{Gabor1961} argued that acquiring information through measurement necessarily requires thermodynamic cost, while Landauer \cite{Landauer1961} and Bennett \cite{Bennett1982} argued that erasing measurement outcomes also requires thermodynamic cost.
Sagawa and Ueda \cite{Sagawa2008,Sagawa2009,Sagawa2011} established a quantitative and unified formulation of the second law of information thermodynamics.
They first showed that, in thermodynamic processes involving measurements, the upper bound on extractable work acquires an additional contribution given by the mutual information \cite{Sagawa2008}.
They further derived the thermodynamic costs required for information acquisition and erasure, and showed that the overall process remains consistent with the second law of thermodynamics \cite{Sagawa2009,Sagawa2011}.
Although the original formulation of Sagawa and Ueda was restricted to a limited class of measurement processes, subsequent works \cite{Jacobs2009,Abdelkhalek2018,Minagawa2025} generalized the results to arbitrary quantum measurement processes \footnote{
Consider an indirect measurement process on a target system \(A\) implemented using a memory (or demon) system \(M\).
As explicitly stated in the Erratum by Sagawa and Ueda \cite{Sagawa2011}, the original works of Sagawa and Ueda \cite{Sagawa2008,Sagawa2009,Sagawa2011} were restricted to so-called efficient measurement processes for the target system \(A\).
Jacobs \cite{Jacobs2009} and Abdelkhalek \textit{et al.}\ \cite{Abdelkhalek2018} generalized, respectively, the result on extractable work from the target system \cite{Sagawa2008} and the result on the energetic costs of measurement and information erasure \cite{Sagawa2009,Sagawa2011} to arbitrary measurement processes on \(A\).
In these generalizations, the relevant information quantity is replaced from the QC mutual information \cite{Sagawa2008} to the Groenewold--Ozawa information gain \cite{Groenewold1971,Ozawa1986}.
Minagawa \textit{et al.}\ \cite{Minagawa2025} further generalized these results to situations in which the measurement performed on the memory system \(M\) is itself an arbitrary measurement process rather than a projective one.
Since, in GPTs, arbitrary measurements cannot necessarily be realized as indirect measurements implemented via projective measurements on a memory system (indeed, even the notion of projective measurement is generally nontrivial), the generalization in \cite{Minagawa2025} is particularly important for extending information thermodynamics to GPTs.
}.

However, despite these developments, information thermodynamics has not yet been formulated in a general manner within the framework of GPTs.
This constitutes a fundamental gap in our understanding of the relationship between thermodynamics and measurements in GPTs.

Moreover, the thermodynamic cost of measurements is also an important issue in the context of the SPM model.
Previous studies based on SPM models have assumed that a certain class of measurement processes implemented by SPMs can be performed without any thermodynamic cost.
However, this assumption already involves nontrivial issues even within quantum theory.
Based on physical intuition, von Neumann \cite{Neumann1932} argued that projective measurement processes that do not disturb the initial state can be carried out without thermodynamic cost, although the precise scope and rigorous justification of this claim remain unclear.
Furthermore, in GPTs, even the definition of a class of measurements analogous to projective measurements in quantum theory is generally unclear \cite{Mielnik1969,Araki1980,Gudder1999,Kleinmann2014,Umekawa2026}.
It is therefore even less evident which measurement processes, if any, can be implemented by SPMs without thermodynamic cost.

To address these issues, in this work we formulate general processes involving measurement, feedback, and information erasure in GPTs, thereby constructing a framework of information thermodynamics in GPTs.
More specifically, we derive a generalized second law of information thermodynamics for arbitrary GPT systems and show that, for any subadditive entropy, no excess work can be extracted if the measurement processes satisfy entropy nondecrease.
This result can be regarded as a generalization of the Sagawa--Ueda inequalities \cite{Sagawa2008,Sagawa2009,Sagawa2011,Jacobs2009,Abdelkhalek2018,Minagawa2025} in quantum information thermodynamics to the GPT setting.

Furthermore, we formulate a general framework for isothermal processes in SPM models and obtain a unified understanding of such processes from the viewpoint of the thermodynamic cost of measurements.
All processes considered in previous works \cite{Neumann1927,Neumann1932,Petz2001,Hanggi2013,Barnum2015,Krumm2015,Krumm2017,Minagawa2022} can be regarded as special cases of the general framework introduced in this work.
By analyzing this general class of processes, we show, similarly to the result in information thermodynamics, that no work can be extracted from an isothermal SPM cycle whenever the measurement processes implemented by SPMs satisfy entropy nondecrease with respect to a subadditive entropy.
These results demonstrate that, even for thermodynamic processes involving measurements in GPT systems, the extractable work is constrained by entropy changes, just as in ordinary thermodynamics.
Moreover, whereas most previous studies were restricted to limited classes of GPT systems or particular thermodynamic cycles, our results apply to arbitrary GPT systems and therefore possess a high degree of generality.

Several definitions of entropy in GPTs have been proposed, including the measurement entropy \cite{Short2010,Barnum2010,Kimura2010} and the accessible information entropy \cite{Kimura2010}.
In this work, we derive sufficient conditions for measurement processes to be consistent with entropy nondecrease for these entropies, and show that such processes form candidates for the class of measurement processes that can be implemented by SPMs without thermodynamic cost.
In particular, since projective measurement processes in quantum systems satisfy entropy nondecrease of the von Neumann entropy, we show that no work can be extracted from isothermal SPM cycles employing projective measurements, thereby providing a rigorous justification of von Neumann's original intuitive argument \cite{Neumann1932}.

Finally, we construct and analyze, for the first time, explicit examples of GPT systems and isothermal SPM cycles from which positive work can be extracted.
Although previous work \cite{Minagawa2022} showed that positive work extraction can occur in SPM cycles of GPT systems satisfying certain assumptions, no concrete examples of such systems had been presented.
In this work, we construct explicit SPM cycles realizing positive work extraction using GPT systems with square and regular-hexagon state spaces, and analyze the origin of the extracted work from the viewpoint of entropy changes associated with the measurement processes.
By comparing these examples with the conditions for entropy-nondecreasing measurement processes derived in this work, we show that, in GPTs, the existence of SPMs implementing measurements without thermodynamic cost can lead to violations of the second law due either to the failure of fundamental entropy properties or to discrepancies between different entropy notions.

These results provide a unified framework for understanding the relationship between thermodynamics and measurements in GPTs from a general perspective independent of particular models or physical theories.
In particular, this work explicitly incorporates the thermodynamic cost of measurement processes, which has not been sufficiently taken into account in previous studies, and derives general conditions under which measurement processes are consistent with the second law of thermodynamics.
As a consequence, our framework not only unifies previous results but also offers a general basis for discussing thermodynamic consistency in GPTs.
Furthermore, our results show that, even in GPTs, the amount of extractable work is constrained by entropy changes associated with measurement processes.
In this sense, our work clarifies, within the most general framework of probabilistic physical theories, the correspondence between the operational formulation of the second law in terms of work extraction and its axiomatic formulation in terms of entropy nondecrease.

\section{Preliminaries}

We review the framework of generalized probabilistic theories (GPTs) \cite{Ludwig1964,Davies1970,Gudder1973,Barrett2007,Janotta2014,Plavala2023} and introduce notations used in this paper.
The framework of GPTs is based on minimal assumptions for the probabilistic predictability of measurement outcomes.
Thus, it offers a framework for describing the landscape of possible theories of the world including classical and quantum theories, as well as many other hypothetical theories beyond them \cite{Popescu1994,Janotta2011,Jordan1934,Barnum2020,Yoshida2020,Plavala2022,Plavala2023-2,Jiang2024}.
For the sake of brevity, we here introduce the formalism in a top-down manner; however, it should be noted that it can be reconstructed solely from physical requirements.

\subsection{Systems, states, and effects}

A \textit{GPT system} \(A\) is specified by a triple \(\p{V_A,C_A,\mathds1_A}\), where \(V_A\) is a real vector space, \(C_A\subset V_A\) is a positive cone, and \(\mathds1_A:V_A\to\mathds R\) is a unit effect.
A subset \(C_A\) of a real vector space \(V_A\) is called a \textit{positive cone} if it satisfies the following conditions:
\begin{equation}
    \begin{aligned}
        \mbox{(i)}\quad&\mathds{R}_{\geq0}C_A\subseteq C_A,\\
        \mbox{(ii)}\quad&C_A+C_A\subseteq C_A,\\
        \mbox{(iii)}\quad&C_A\cap(-C_A)=\B{0},\\
        \mbox{(iv)}\quad&C_A-C_A=V_A.
    \end{aligned}
\end{equation}
The \textit{unit effect} \(\mathds1_A:V_A\to\mathds R\) is a linear function satisfying \(\mathds1_A\rho>0\) for any \(\rho\in C_A\setminus\B{0}\).
We introduce the norm on \(V_A\) by
\begin{equation}
    \V\rho\coloneq\inf\B{\mathds1_A\p{\rho_++\rho_-}\mid\rho_\pm\in C_A,\:\rho_+-\rho_-=\rho}.
\end{equation}
Note that \(\V\rho=\mathds1_A\rho\) for all \(\rho\in C_A\).
We further require that \(V_A\) be complete and \(C_A\) be closed with respect to this norm.

The \textit{state space} of a GPT system \(A\) is given by
\begin{equation}
    \mathrm{St}\p A\coloneq\B{\rho\in C_A\mid\V\rho=1}.
\end{equation}
An element \(\rho\in C_A\setminus\B{0}\) corresponds to an unnormalized state, whose normalization is given by \(\hat{\rho}\coloneq \frac{\rho}{\V{\rho}}\in \mathrm{St}\p A\).

The \textit{effect space} of a GPT system \(A\) is defined by
\begin{equation}
    \mathrm{Eff}\p A\coloneq C_A^*\cap\p{\mathds1_A-C_A^*},
\end{equation}
where
\begin{equation}
    C_A^*\coloneq\B{e\in V_A^*\mid eC_A=\mathds R_{\ge0}}
\end{equation}
is the \textit{dual cone} of \(C_A\).
An effect assigns to each state the probability of an event.
In particular, the unit effect \(\mathds1_A\in\mathrm{Eff}\p A\) assigns unit probability to every state.

We introduce the partial order on \(V_A\) and \(V_A^*\) by \(\rho\le\sigma\vcentcolon\Leftrightarrow\sigma-\rho\in C_A\) and \(e\le f\vcentcolon\Leftrightarrow f-e\in C_A^*\), respectively.
For unnormalized states \(\rho,\sigma\in C_A\), \(\rho\le \sigma\) implies that the state \(\hat{\sigma}\) can be prepared by mixing the states \(\hat{\rho}\) and \(\hat{\sigma-\rho}\).
For effects \(e,f\in\mathrm{Eff}\p A\), \(e\le f\) implies that the probability of the occurrence of \(f\) is greater than or equal to the probability of the occurrence of \(e\) for any state.
We call a state \(\rho\in\mathrm{St}\p{A}\) to be \textit{pure} if
\begin{equation}
    0\le \sigma\le \rho \implies \exists \lambda\in[0,1]\quad \mbox{s.t.}\quad \sigma=\lambda\rho.
\end{equation}
Similarly, we call an effect \(e\in\mathrm{Eff}\p{A}\) to be \textit{indecomposable} if
\begin{equation}
    0\le f\le e \implies \exists \lambda\in[0,1]\quad \mbox{s.t.}\quad f=\lambda e.
\end{equation}

The most prominent example of a GPT system is a quantum system.
A quantum system specified by a Hilbert space \(\mathscr{H}\) is a GPT system specified by the triple \(\p{V_\mathscr{H}, C_\mathscr{H},\mathds1_\mathscr{H}}\coloneq\p{\mathcal T_\mathrm{sa} \p{\mathscr H},\mathcal T_+\p{\mathscr H},\mathds1_\mathscr{H}}\), where \(\mathcal T_\mathrm{sa}\p{\mathscr H}\) and \(\mathcal T_+\p{\mathscr H}\) denote the space of self-adjoint trace-class operators and positive trace-class operators on \(\mathscr{H}\), and \(\mathds1_{\mathscr H}\rho\coloneq\mathrm{tr}\,\rho\).
Then, \(V_\mathscr{H}^*\) and \(C_\mathscr{H}^*\) are isomorphic to the space of self-adjoint bounded operators \(\mathcal B_\mathrm{sa}\p{\mathscr H}\) and positive operators \(\mathcal B_+ \p{\mathscr H}\) on \(\mathscr H\), and \(\mathds1_{\mathscr H}\) corresponds to the identity operator on \(\mathscr H\) under this isomorphism.
Thus, the state space and effect space of a quantum system are given by
\begin{equation}
    \mathrm{St}\p{\mathscr H}=\B{\rho\in\mathcal T_+\p{\mathscr H}\mid\mathrm{tr}\,\rho=1}
\end{equation}
and
\begin{equation}
    \mathrm{Eff}\p{\mathscr H}\cong\B{e\in\mathcal B_+\p{\mathscr H}\mid e\le\mathds1_\mathscr{H}},
\end{equation}
respectively.

\subsection{Processes and composite systems}

Processes on GPT systems are described by positive maps.
For GPT systems \(A\) and \(B\), a linear map \(\mathcal F:V_A\to V_B\) satisfying \(\mathcal F\p{C_A}\subseteq C_B\) is called a \textit{positive map}, and is denoted by \(\mathcal F:A\to B\).
A positive map \(\mathcal F:A\to B\) satisfying \(\mathds1_B\mathcal F=\mathds1_A\) is called a \textit{process}.
This ensures that state normalization is preserved, which corresponds to the conservation of probability.
The composition \(\mathcal G\mathcal F:A\to C\) of processes \(\mathcal F:A\to B\) and \(\mathcal G:B\to C\) is also a process, representing the sequential application of \(\mathcal F\) followed by \(\mathcal G\).
The process given by the identity map on \(V_A\) is denoted by \(\mathrm{id}_A\).
A process that admits an inverse map that is also a process is called \textit{reversible}.

In the following, we assume that all GPT systems are finite-dimensional.
A GPT system \(AB\) is called a \textit{composite system} of two systems \(A\) and \(B\) if it satisfies
\begin{equation}
    \begin{aligned}
        V_{AB}&=V_A\otimes V_B,\\
        C_{AB}&\supseteq C_A\otimes C_B,\\
        C_{AB}^*&\supseteq C_A^*\otimes C_B^*,\\
        \mathds1_{AB}&=\mathds1_A\otimes\mathds1_B.
    \end{aligned}
\end{equation}
For composite systems \(A_1B_1\) and \(A_2B_2\), and processes \(\mathcal F:A_1\to A_2\) and \(\mathcal G:B_1\to B_2\), we require that the tensor product map \(\mathcal F\otimes\mathcal G:V_{A_1}\otimes V_{B_1}\to V_{A_2}\otimes V_{B_2}\) is a process, i.e., \(\mathcal F\otimes\mathcal G:A_1B_1\to A_2B_2\), representing the parallel composition of \(\mathcal F\) and \(\mathcal G\).
This requirement may restrict the set of allowed processes on the component systems, analogous to the requirement of complete positivity in quantum theory.
It is important to note that the definition of composite systems is not unique in general.
In particular, one can introduce the \textit{minimal tensor product} \(C_A\otimes_{\min}C_B\) and \textit{maximal tensor product} \(C_A\otimes_{\max}C_B\), defined by
\begin{equation}
    \begin{aligned}
        C_A\otimes_{\min}C_B&\coloneq\mathrm{conv}\p{C_A\otimes C_B},\\
        \p{C_A\otimes_{\max}C_B}^*&\coloneq \mathrm{conv}\p{C_A^*\otimes C_B^*},
    \end{aligned}
\end{equation}
as well as intermediate constructions lying between them.
For example, for two quantum systems with Hilbert spaces \(\mathscr H_1\) and \(\mathscr H_2\), one has
\begin{equation}
    \begin{aligned}
        \mathcal T_+\p{\mathscr H_1}\otimes_{\min}\mathcal T_+\p{\mathscr H_2}&\subsetneq\mathcal T_+\p{\mathscr H_1\otimes\mathscr H_2}\\
        &\subsetneq\mathcal T_+\p{\mathscr H_1}\otimes_{\max}\mathcal T_+\p{\mathscr H_2}.
    \end{aligned}
\end{equation}
Thus, the standard notion of a composite system in quantum theory does not coincide with either the minimal or the maximal tensor product.

For a state \(\rho^{AB}\in\mathrm{St}\p{AB}\), the reduced state on subsystem \(A\) is defined by \(\rho^A\coloneq\mathrm{tr}_B\rho^{AB}\coloneq\p{\mathrm{id}_A\otimes\mathds1_B}\rho^{AB}\).
The partial trace \(\mathrm{tr}_{B}:AB\to A\) is a process.

\subsection{Classical systems and direct sums}

A \textit{classical system} specified by a finite set \(X\) is defined as the GPT system \(\p{V_X,C_X,\mathds1_X}\), where
\begin{equation}
    \begin{aligned}
        V_X&\coloneq\mathds R^X=\B{\p{p_x}_{x\in X}\mid\forall x\in X,\:p_x\in\mathds R},\\
        C_X&\coloneq\mathds R_{\ge0}^X=\B{\p{p_x}_{x\in X}\mid\forall x\in X,\:p_x\in\mathds R_{\ge0}},\\
        \mathds1_X&\p{p_x}_{x\in X}\coloneq\sum_{x\in X}p_x.
    \end{aligned}
\end{equation}
The state space of the classical system is the simplex
\begin{equation}
    \mathrm{St}\p X=\B{\p{p_x}_{x\in X}\in\mathds R^X_{\ge0}\mid\sum_{x\in X}p_x=1},
\end{equation}
with pure states given by
\begin{equation}
    \delta_x\coloneq\p{\delta_{x,x'}}_{x'\in X},\quad x\in X.
\end{equation}
Any state can be written uniquely as a probabilistic mixture of these pure states: \(\p{p_x}_{x\in X}=\sum_{x\in X}p_x\delta_x\).
We denote by \(\epsilon_x\) the effect that returns the probability of being in the state \(\delta_x\), namely \(\epsilon_x\p{p_x}_{x\in X}\coloneq p_x\).
With this notation, the unit effect is expressed as \(\mathds1_X=\sum_{x\in X}\epsilon_x\).

We next introduce the direct sum of systems.
For a family of systems \(\p{A_x}_{x\in X}\) indexed by a classical system \(X\), we define their \textit{direct sum} \(\bigoplus_{x\in X}A_x\) as the system specified by
\begin{equation}
    \begin{aligned}
        V_{\bigoplus_{x\in X}A_x}&\coloneq\B{\p{\rho_x^{A_x}}_{x\in X}\mid\forall x\in X,\:\rho_x^{A_x}\in V_{A_x}},\\
        C_{\bigoplus_{x\in X}A_x}&\coloneq\B{\p{\rho_x^{A_x}}_{x\in X}\mid\forall x\in X,\:\rho_x^{A_x}\in C_{A_x}},\\
        \mathds1_{\bigoplus_{x\in X}A_x}&\p{\rho_x^{A_x}}_{x\in X}\coloneq\sum_{x\in X}\mathds1_{A_x}\rho_x.
    \end{aligned}
\end{equation}
The direct sum describes the statistical mixture of different systems.
We denote by
\begin{equation}
    \aab{g_x}_{x\in X}\p{\p{\rho_x^{A_x}}_{x\in X}}\coloneq\sum_{x\in X}\V{\rho_x^{A_x}}g_x\p{\hat\rho_x^{A_x}}
\end{equation}
the expectation value of a family of functions \(\p{g_x:\mathrm{St}\p{A_x}\to\mathds R}_{x\in X}\) evaluated at the state \(\p{\rho_x^{A_x}}_{x\in X}\in\mathrm{St}\p{\bigoplus_{x\in X}A_x}\).

The composition rule of two GPT systems is uniquely determined, \textit{i.e.}, the minimal and maximal tensor products coincide with each other, if and only if at least one of them is classical \cite{Aubrun2021,Aubrun2022}.
The composite system of a classical system \(X\) and a GPT system \(A\) is given by the direct sum \(XA\coloneq\bigoplus_{x\in X}A\).
A state of the system \(XA\) describes an ensemble of states of the system \(A\).
We write \(\rho^{XA}\in\mathrm{Ens}\p{\rho^A}\) when the state \(\rho^{XA}=\p{\rho_x^A}_{x\in X}\) of the system \(XA\) satisfies \(\sum_{x\in X}\rho_x^A=\rho^A\).
Furthermore, we write \(\rho^{XA}\in\mathrm{Ens}_\mathrm{pure}\p\rho^A\) if each \(\hat\rho_x^A\) is pure in addition.

\subsection{Measurements and measurement processes}

An \(X\)-outcome \textit{measurement} on a GPT system \(A\) is defined as a process \(\mathcal{E}:A\to X\).
A measurement \(\mathcal{E}:A\to X\) is in one-to-one correspondence with a family of effects \(\p{e_x\in\mathrm{Eff}\p A}_{x\in X}\) satisfying \(\sum_{x\in X}e_x=\mathds1_A\), where the correspondence is given by \(e_x=\epsilon_x\mathcal E\).
When performing the measurement \(\mathcal{E}\) on the state \(\rho\), \(e_x\rho\) gives the probability of obtaining the outcome \(x\in X\).
A measurement \(\p{e_x}_{x\in X}\) is said to be \textit{fine-grained} if each \(e_x\) is indecomposable.
We denote by \(\mathrm{Meas}\p A\) and \(\mathrm{Meas}_{\mathrm{fg}}\p A\) the class of all measurements and all fine-grained measurements on \(A\), respectively.

An \(X\)-outcome \textit{measurement process} on a GPT system \(A\) is defined as a process \(\mathcal{M}:A\to XA\).
For such a process \(\mathcal M\), \(\mathrm{tr}_A\mathcal M\) defines a measurement.
A measurement process \(\mathcal{M}:A\to XA\) is in one-to-one correspondence with a family of positive maps \(\p{\mathcal M_x:A\to A}_{x\in X}\) satisfying \(\sum_{x\in X}\mathds1_A\mathcal M_x=\mathds1_A\), where the correspondence is given by \(\mathcal M_x=\p{\mathrm{id}_A\otimes\epsilon_x}\mathcal M\).
When performing the measurement process \(\mathcal{M}\) on the state \(\rho\), \(\V{\mathcal{M}_x\rho}\) and \(\hat{\mathcal M_x\rho}\) describe the probability of obtaining the outcome \(x\in X\) and the post-measurement state conditioned on that outcome, respectively.
A measurement process \(\p{\mathcal M_x}_{x\in X}\) is said to be \textit{repeatable} if \(\mathcal M_x\mathcal M_{x'}=\delta_{x,x'}\mathcal M_x\) holds for all \(x,x'\in X\).

\section{Entropies in generalized probabilistic theories}

In this section, we introduce several entropies in GPTs and investigate their properties, particularly in relation to measurement processes and entropy nondecrease.
Proofs of the statements in this section are provided in Appendix \ref{app:entropies}.

\subsection{Definitions and properties of entropies}
\label{Definitions and properties of entropies}

Entropy plays a fundamental role in thermodynamics.
However, defining entropy in GPTs is nontrivial, and various generalizations of the von Neumann entropy in quantum theory have been proposed \cite{Short2010,Barnum2010,Kimura2010,Kimura2016}.

We first introduce various proposed definitions of entropies in GPTs.
An \textit{entropy} \(S\) is a family of functions \(S^A:\mathrm{St}\p A\to\mathds R\) defined for each system \(A\) in some class of systems.

We denote by \(H^X\p{\p{p_x}_{x\in X}}\) the Shannon entropy of a probability distribution \(\p{p_x}_{x\in X}\) and by \(I^{X:Y}\p{\p{p_{x,y}}_{x\in X,y\in Y}}\) the mutual information between \(X\) and \(Y\) for the probability distribution \(\p{p_{x,y}}_{x\in X,y\in Y}\).
\begin{dfn}
    \leavevmode
    \begin{itemize}
        \item The \textit{mixing entropy} \cite{Short2010,Barnum2010,Kimura2010,Kimura2016} \(S_\mathrm{mix}\) is defined by
        \begin{equation}
            S_\mathrm{mix}^A\p{\rho^A}\coloneq\inf_{\rho^{XA}\in\mathrm{Ens}_{\mathrm{pure}}\p{\rho^A}}H^X\p{\mathrm{tr}_A\rho^{XA}}.
        \end{equation}
        \item The \textit{measurement entropy} \cite{Short2010,Barnum2010,Kimura2010,Kimura2016} \(S_\mathrm{meas}\) is defined by
        \begin{equation}
            S_\mathrm{meas}^A\p{\rho^A}\coloneq\inf_{\p{\mathcal E:A\to X}\in\mathrm{Meas}_\mathrm{fg}\p A}H^X\p{\mathcal E\rho^A}.
        \end{equation}
        \item The \textit{accessible information entropy} \cite{Kimura2010,Kimura2016} \(S_\mathrm{acc}\) is defined by
        \begin{equation}
            S_\mathrm{acc}^A\p{\rho^A}\coloneq\sup_{\rho^{XA}\in\mathrm{Ens}\p{\rho}}I_\mathrm{acc}^{X:A}\p{\rho^{XA}},
        \end{equation}
        where
        \begin{equation}
            I_\mathrm{acc}^{X:A}\p{\rho^{XA}}\coloneq\sup_{\p{\mathcal E:A\to Y}\in\mathrm{Meas}\p A}I^{X:Y}\p{\p{\mathrm{id}_X\otimes\mathcal E}\rho^{XA}}.
        \end{equation}
        \item The \textit{induced entropy} \cite{Kimura2016} \(S'\) of an entropy \(S\) is defined by
        \begin{equation}
            {S'}^A\p{\rho^A}\coloneq\sup_{\rho^{XA}\in\mathrm{Ens}(\rho)}\p{I_\mathrm{acc}^{X:A}\p{\rho^{XA}}+\aab{S^A}_X\p{\rho^{XA}}}.
        \end{equation}
        \item The \textit{infinity entropy} \cite{Kimura2016} \(S^\infty\) of an entropy \(S\) is defined by
        \begin{equation}
            {S^\infty}^A\p{\rho^A}\coloneq\sup\B{S^A\p{\rho^A},{S'}^A\p{\rho^A},{S''}^A\p{\rho^A},\dots}.
        \end{equation}
        \item For an entropy \(S^A\) defined for a system \(A\), its \textit{extended entropy} is defined for composite systems \(XA\) with a classical system \(X\) by
        \begin{equation}
            \tilde S^{XA}\p{\rho^{XA}}\coloneq H\p{\mathrm{tr}_A\rho^{XA}}+\aab{S^A}_X\p{\rho^{XA}}.
        \end{equation}
    \end{itemize}
\end{dfn}

The entropies \(S_\mathrm{mix}\), \(S_\mathrm{meas}\), and \(S_\mathrm{acc}\) all reduce to the von Neumann entropy for quantum systems and to the Shannon entropy for classical systems \cite{Short2010,Barnum2010,Kimura2010,Kimura2016}.
Moreover, when \(S\) is the von Neumann (respectively, Shannon) entropy, \(S'\) and \(S^\infty\) also coincide with the von Neumann (respectively, Shannon) entropy \cite{Kimura2016}.
Furthermore, when \(S\) is the Shannon entropy, \(\tilde S\) also coincides with the Shannon entropy.

We next introduce several desirable properties of entropies, which are not necessarily satisfied by all entropies.
\begin{dfn}
    \leavevmode
    \begin{itemize}
        \item An entropy \(S\) is said to be \textit{invariant under reversible processes} if it satisfies \(S^A\p{\mathcal U\rho^A}=S^A\p{\rho^A}\)
        for any state \(\rho^A\) and any reversible process \(\mathcal U\).
        \item An entropy \(S\) is said to be \textit{concave} if it satisfies
        \begin{equation}
            S^A\p{\mathrm{tr}_X\rho^{XA}}\ge\aab{S^A}_X\p{\rho^{XA}}
        \end{equation}
        for any classical system \(X\) and any state \(\rho^{XA}\in\mathrm{St}\p{XA}\).
        \item An entropy \(S\) is said to satisfy the \textit{Holevo bound} if \(S'=S\) holds \cite{Kimura2016}.
        \item An entropy \(S\) is said to be \textit{subadditive} if it satisfies
        \begin{equation}
        \label{eq:subadditivity}
            S^{AB}\p{\rho^{AB}}\le S^A\p{\mathrm{tr}_B\rho^{AB}}+S^B\p{\mathrm{tr}_A\rho^{AB}}
        \end{equation}
        for any state \(\rho^{AB}\in\mathrm{St}\p{AB}\), with equality for every product state \(\rho^{AB}=\rho^A\otimes\rho^B\).
        \item An entorpy \(S\) is said to be \textit{consistent with the direct sum} if it coincides with the Shannon entropy for classical systems and satisfies
        \begin{equation}
            \begin{aligned}
                &S^{\bigoplus_{x\in X}A_x}\p{\p{\rho_x^{A_x}}_{x\in X}}\\
                &=H^X\p{\p{\V{\rho_x^{A_x}}}_{x\in X}}+\aab{S^{A_x}}_{x\in X}\p{\p{\rho_x^{A_x}}_{x\in X}}
            \end{aligned}
        \end{equation}
        for any ensemble of states \(\p{\rho_x^{A_x}}_{x\in X}\in\mathrm{St}\p{\bigoplus_{x\in X}A_x}\).
    \end{itemize}
\end{dfn}
The von Neumann entropy and the Shannon entropy satisfy all of these properties.

The entropies \(S_\mathrm{mix}\), \(S_\mathrm{meas}\), and \(S_\mathrm{acc}\) are invariant under reversible processes.
Moreover, the induced entropy, infinity entropy, and extended entropy of an entropy that is invariant under reversible processes are also invariant under reversible processes.

The measurement entropy and entropies satisfying the Holevo bound are concave \cite{Short2010,Barnum2010,Kimura2010,Kimura2016}.
An infinite entropy is conjectured to satisfy the Holevo bound \cite{Kimura2016}.

The measurement entropy always satisfies Eq.~\eqref{eq:subadditivity} because the tensor product of two fine-grained measurements is again fine-grained, as follows from Proposition 5.7 of Ref.~\cite{Debruyn2022}.
For the maximal tensor product, the measurement entropy is moreover additive on product states and is therefore subadditive in the stronger sense adopted here.
This includes the canonical composition of an arbitrary system with a classical system.
Additivity on product states also holds for the standard tensor product of quantum systems.
However, the additivity can fail for a general composite system.

In settings where all but one of the systems are classical, concavity and subadditivity are equivalent for an entropy \(S\) that is consistent with the direct sum.
We also find the following:
\begin{prop}
    \label{prop:consistency with the direct sum}
    \leavevmode
    \begin{enumerate}
        \item
        \(S_\mathrm{meas}\) is consistent with the direct sum.
        \item
        \(S_\mathrm{acc}\) is consistent with the direct sum.
        \item
        If an entropy \(S\) is consistent with the direct sum, then \(S'\) and \(S^\infty\) are also consistent with the direct sum.
    \end{enumerate}
\end{prop}

\subsection{Measurement processes satisfying entropy nondecrease}

In this subsection, we consider measurement processes satisfying entropy nondecrease.
As we will see in later sections, whether processes involving measurements decrease entropies becomes a central issue in the consistency between the second law of thermodynamics and physical processes in GPTs.
Here, we provide sufficient conditions under which a measurement process and the subsequent forgetting of its outcome do not decrease the entropies introduced above.

We first discuss conditions under which a measurement process does not decrease an entropy.
We obtain the following result.
\begin{prop}
    \label{prop:measurement entropy non-decresing measurement processes}
    For any repeatable measurement process \(\mathcal M:A\to XA\) such that \(\mathrm{tr}_A \mathcal M\) is fine-grained,
    \begin{equation}
        S_\mathrm{meas}^A\p{\rho^A}\le S_\mathrm{meas}^{XA}\p{\mathcal M\rho^A}
    \end{equation}
    holds for any state \(\rho^A\in\mathrm{St}\p A\).
\end{prop}
In particular, this implies that any rank-\(1\) projective measurement process in quantum systems do not decrease the von Neumann entropy.
In fact, this remains true for arbitrary L\"uders measurement processes \(\rho\mapsto\p{\Pi_x\rho\Pi_x}_{x\in X}\) in quantum systems, where \(\p{\Pi_x}_{x\in X}\) is a family of mutually orthogonal projection operators, not necessarily of rank 1.

We next discuss conditions under which an entropy does not decrease when the measurement outcome is forgotten.
To this end, we introduce the following notion of \textit{information-preserving processes}.
\begin{dfn}
    An process \(\mathcal F:A\to B\) is said to \textit{preserve the information of an initial state} \(\rho^A\) if, for any \(\rho^{XA}\in\mathrm{Ens}\p{\rho^A}\), there exists a process \(\mathcal G:B\to A\) such that \(\p{\mathrm{id}_X\otimes\mathcal G\mathcal F}\rho^{XA}=\rho^{XA}\).
\end{dfn}
This definition means that the information of the classical system \(X\) encoded in system \(A\) can be perfectly recovered by \(\mathcal G\), even after applying \(\mathcal F\).
A reversible process on system \(A\) preserves the information of every initial state.
In addition, if a process \(\mathcal F:A\to B\) preserves the information of an initial state \(\rho^A\), then for any \(0<\sigma^A\le\rho^A\), \(\mathcal F\) also preserves the information of \(\hat\sigma^A\).

We say that an entropy \(S\) is \textit{nondecreasing under information-preserving processes} if it satisfies
\begin{equation}
    S^A\p{\rho^A}\le S^B\p{\mathcal F\rho^A}
\end{equation}
for any state \(\rho^A\) and any process \(\mathcal F\) that preserves the information of \(\rho^A\).
We then obtain the following result.
\begin{prop}
    \label{prop:entropy nondecreasing under information-preserving processes}
    If an entropy \(S\) is nondecreasing under information-preserving processes, then \(S'\) and \(S^\infty\) are also nondecreasing under information-preserving processes.
    In particular, \(S_\mathrm{acc}=0'\) is nondecreasing under information-preserving processes.
\end{prop}

We next analyze conditions under which discarding the outcome of a measurement process becomes an information-preserving process.
We introduce the following notion of \textit{strongly repeatable measurement processes}.
\begin{dfn}
    A measurement process \(\mathcal M=\p{\mathcal M_x}_{x\in X}:A\to XA\) is said to be \textit{strongly repeatable} if, for any state \(\rho^A\in\mathrm{St}\p A\) and any \(x\in X\), \(\mathcal M_x\sigma^A=\sigma^A\) holds for every \(0\le\sigma^A\le\mathcal M_x\rho^A\).
\end{dfn}
In particular, L\"uders measurement processes in quantum systems and measure-and-prepare processes that prepare perfectly distinguishable pure states are strongly repeatable.
Here, a family of states \(\p{\rho_x^A\in\mathrm{St}\p A}_{x\in X}\) is said to be \textit{perfectly distinguishable} if there exists a measurement \(\p{e_x}_{x\in X}:A\to X\) such that \(e_x\rho_{x'}^A=\delta_{x,x'}\).
The corresponding measure-and-prepare process is given by \(\p{\rho_x^Ae_x}_{x\in X}:A\to XA\).
By contrast, a measure-and-prepare process that prepares perfectly distinguishable mixed states is not necessarily strongly repeatable.

We then obtain the following characterization.
\begin{prop}
    \label{prop:strogly repeatable implies information preserving}
    A measurement process \(\mathcal M:A\to XA\) is strongly repeatable if and only if, for any state \(\rho^A\in\mathrm{St}\p A\), \(\p{\mathrm{id}_Y\otimes\mathcal M\,\mathrm{tr}_X}\sigma^{YXA}=\sigma^{YXA}\) holds for every \(\sigma^{YXA}\in \mathrm{Ens}\p{\mathcal M\rho^A}\).
    In particular, if \(\mathcal M\) is strongly repeatable, then \(\mathrm{tr}_X:XA\to A\) preserves the information of the initial state \(\mathcal M\rho^A\) for any state \(\rho^A\in\mathrm{St}\p A\).
\end{prop}
Proposition \ref{prop:strogly repeatable implies information preserving} shows that, for strongly repeatable measurement processes, the information of the measurement outcome can be recovered by performing the same measurement again.
Combining Proposition \ref{prop:entropy nondecreasing under information-preserving processes} and Proposition \ref{prop:strogly repeatable implies information preserving}, we conclude that forgetting the outcome of a strongly repeatable measurement does not decrease \(S_\mathrm{acc}\).

\section{Information thermodynamics in generalized probabilistic theories}

We consider a generalization of information thermodynamics \cite{Sagawa2008,Sagawa2009,Sagawa2011,Jacobs2009,Abdelkhalek2018,Minagawa2025} to GPTs.
In particular, the framework developed in this section generalizes the setting and results of Ref.\ \cite{Minagawa2025} from quantum theory to arbitrary GPTs.

Let the internal energy of a GPT system \(A\) be given by an affine function \(E^A:\mathrm{St}\p A\to\mathds R\).
For an entropy \(S^A\) and inverse temperature \(\beta\), we define the \textit{nonequilibrium free energy} of system \(A\) by \(F_\beta^A\p{\rho^A}:=E^A\p\rho-\beta^{-1}S^A\p{\rho^A}\).
We also define \(\gamma_\beta^A\) as the state that minimizes the nonequilibrium free energy at inverse temperature \(\beta\).
When \(A\) is a classical (respectively, quantum) system, \(\gamma_\beta^A\) coincides with the canonical distribution (respectively, Gibbs state).

We consider a composite system consisting of a target system \(A\), a memory device \(M\), and heat baths \(B_1\), \(B_2\), \(B_3\).
We assume that the internal energy of the total system is given by
\begin{equation}
    E^{AMB_1B_2B_3}=E^A+E^M+E^{B_1}+E^{B_2}+E^{B_3}.
\end{equation}
We consider the process illustrated in Fig.~\ref{fig:information thermodynamics}.

\begin{figure}
    \centering
    \begin{tikzpicture}
        \draw[line width=1pt](-0.25,0)--(5.25,0);
        \draw(-0.25,0)node[anchor=south west]{\(\scriptstyle\!M\)};
        \draw[line width=1pt](-0.25,2.4)--(5.25,2.4);
        \draw(-0.25,2.4)node[anchor=south west]{\(\scriptstyle\!A\)};
        \draw[line width=0.8pt,fill=black!10](0.5,-0.5)rectangle(1.5,2.9);
        \draw(1,1.2)node{\(\mathcal M\)};
        \draw[line width=1pt](1.5,1.2)node[anchor=south west]{\(\scriptstyle\!k\in K\)}--(3.5,1.2);
        \draw(1.5,0)node[anchor=south west]{\(\scriptstyle\!M\)};
        \draw(1.5,2.4)node[anchor=south west]{\(\scriptstyle\!A\)};
        \draw[line width=0.8pt,fill=black!10](2,1.9)rectangle(3,2.9);
        \draw(2.5,2.4)node{\(\mathcal F_k\)};
        \draw[line width=1pt](2.5,1.2)--(2.5,1.9);
        \fill(2.5,1.2)circle[radius=3pt];
        \draw(3,2.4)node[anchor=south west]{\(\scriptstyle\!A\)};
        \draw[line width=0.8pt,fill=black!10](3.5,-0.5)rectangle(4.5,1.7);
        \draw(4,0.6)node{\(\mathcal V\)};
        \draw(4.5,0)node[anchor=south west]{\(\scriptstyle\!M\)};
    \end{tikzpicture}
    \caption{Process involving measurement, feedback, and information erasure considered for the derivation of the generalized Sagawa–Ueda inequalities. The heat baths are omitted from the figure.}
    \label{fig:information thermodynamics}
\end{figure}
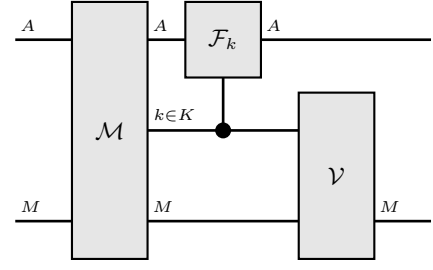

\begin{enumerate}
    \item The initial state is given by
    \begin{equation}
        \rho_{\p0}^{AMB_1B_2B_3}\coloneq\rho_{\p0}^A\otimes\rho_{\p0}^M\otimes\gamma_\beta^{B_1}\otimes\gamma_\beta^{B_2}\otimes\gamma_\beta^{B_3}.
    \end{equation}
    \item A measurement process \(\mathcal M:AMB_1\to KAMB_1\) is performed, where \(K\) is a classical system representing the measurement outcome.
    We additionally require that \(\mathcal M\otimes\mathrm{id}_{B_2B_3}\) is a valid process.
    (The same convention will be adopted below.)
    We denote the resulting state by \(\rho_{\p1}^{KAMB_1B_2B_3}\).
    \item A feedback process \(\mathcal F_k:AB_2\to AB_2\) depending on the measurement outcome \(k\in K\) is performed.
    More precisely, we apply the process \(\mathcal F=\sum_{k\in K}\p{\mathcal F_k\otimes\delta_k\epsilon_k}\) on system \(KAB_2\).
    We denote the resulting state by \(\rho_{\p2}^{KAMB_1B_2B_3}\).
    \item Finally, a process \(\mathcal V:KMB_3\to MB_3\) is performed to erase the measurement outcome.
    We denote the resulting state by \(\rho_{\p3}^{AMB_1B_2B_3}\).
    After this process, the state of the memory device returns to its initial state, namely, \(\rho_{\p3}^M=\rho_{\p0}^M\), where we use the shorthand notation \(\rho_{\p3}^M\coloneq \mathrm{tr}_{AB_1B_2B_3}\rho_{\p3}^{AMB_1B_2B_3}\), and similarly for other reduced states.
\end{enumerate}

We define the extracted work \(W_\mathrm{ext}\), the work cost of measurement \(W_\mathrm{meas}\), the work cost of erasure \(W_\mathrm{eras}\) and their total contribution \(W\) by
\begin{equation}
    \begin{aligned}
        W_\mathrm{ext}&\coloneq-E^{AB_2}\p{\rho_{\p2}^{AB_2}}+E^{AB_2}\p{\rho^A_{\p0}\otimes \gamma^{B_2}_\beta},\\
        W_\mathrm{meas}&\coloneq E^{MB_1}\p{\rho_{\p1}^{MB_1}}-E^{MB_1}\p{\rho^M_{\p0}\otimes \gamma^{B_1}_{\beta}},\\
        W_\mathrm{eras}&\coloneq E^{MB_3}\p{\rho_{\p3}^{MB_3}}-E^{MB_3}\p{\rho_{\p2}^M\otimes\gamma^{B_3}_\beta},\\
        W&\coloneq E^{AMB_1B_2B_3}\p{\rho_{\p3}^{AMB_1B_2B_3}}\\
        &\qquad-E^{AMB_1B_2B_3}\p{\rho_{\p0}^{AMB_1B_2B_3}}\\
        &=W_\mathrm{ext}-W_\mathrm{meas}-W_\mathrm{eras}.
    \end{aligned}
\end{equation}
Furthermore, we define the Shannon entropy of the measurement outcome by
\begin{equation}
    H\coloneq H^K\p{\rho_{\p1}^K},
\end{equation}
and, for an entropy \(S\), the Groenewold–Ozawa information gain \cite{Groenewold1971,Ozawa1986} by
\begin{equation}
    I\coloneq S^A\p{\rho^A_{\p0}}-\aab{S^A}_K\p{\rho_{\p1}^{KA}}.
\end{equation}
We also define the changes in nonequilibrium free energy by
\begin{equation}
    \begin{aligned}
        \Delta F_\beta^A&\coloneq F_\beta^A\p{\rho_{\p2}^A}-F_\beta^A\p{\rho_{\p0}^A},\\
        \aab{\Delta F_\beta^M}&\coloneq\aab{F_\beta^M}_K\p{\rho_{\p1}^{KM}}-F_\beta^M\p{\rho_{\p0}^M}.
    \end{aligned}
\end{equation}
Finally, we define the entropy production terms by
\begin{equation}
    \begin{aligned}
        \Delta S_\mathcal{M}&\coloneq\tilde S^{KAMB_1}\p{\rho_{\p1}^{KAMB_1}}\\
        &\qquad-S^{AMB_1}\p{\rho_{\p0}^A\otimes\rho_{\p0}^M\otimes\gamma^{B_1}_{\beta}},\\
        \Delta S_\mathcal{F}&\coloneq\tilde S^{KAB_2}\p{\rho_{\p2}^{KAB_2}}-\tilde S^{KAB_2}\p{\rho_{\p1}^{KA}\otimes\gamma^{B_2}_{\beta}},\\
        \Delta S_\mathcal{V}&\coloneq S^{MB_3}\p{\rho_{\p3}^{MB_3}}-\tilde S^{KMB_3}\p{\rho_{\p2}^{KM}\otimes\gamma^{B_3}_{\beta}}.
    \end{aligned}
\end{equation}

We obtain the following \textit{generalized Sagawa--Ueda inequalities} in GPTs
\begin{thm}
    \label{thm:Sagawa--Ueda}
    If the entropy \(S\) is subadditive, then
    \begin{equation}
        \begin{aligned}
            W_\mathrm{ext}&\le-\Delta F_\beta^A+\beta^{-1}\p{I-\Delta S_\mathcal{F}},\\
            W_\mathrm{meas}&\ge\aab{\Delta F_\beta^M}-\beta^{-1}\p{H-I-\Delta S_\mathcal{M}},\\
            W_\mathrm{eras}&\ge-\aab{\Delta F_\beta^M}+\beta^{-1}\p{H+\Delta S_\mathcal{V}}.
        \end{aligned}
    \end{equation}
    Consequently, \(W\le-\Delta F_\beta^A-\beta^{-1}\p{\Delta S_\mathcal{M}+\Delta S_\mathcal{F}+\Delta S_\mathcal{V}}\).
\end{thm}

The proof is provided in Appendix \ref{app:Sagawa--Ueda}.

\begin{rmk}
    \label{rmk:Sagawa--Ueda}
    The measurement entropy satisfies the assumptions of Theorem \ref{thm:Sagawa--Ueda} for the maximal tensor product and the standard tensor product of quantum systems.
    Moreover, when all but one of the systems under consideration are classical, the extended entropy associated with any concave entropy of the one nonclassical subsystem satisfies the assumptions of Theorem \ref{thm:Sagawa--Ueda}.
\end{rmk}

Now consider the case where \(\mathcal F\) and \(\mathcal V\) are reversible and \(\mathcal M=\p{\mathcal M_0\otimes\mathrm{id}_{AB_1}}\mathcal U\), where \(\mathcal U:AMB_1\to AMB_1\) is a reversible process and \(\mathcal M_0:M\to KM\) is a measurement process such that \(\mathcal M_0\otimes\mathrm{id}_{AB_1}\) does not decrease \(S\).
If the entropy \(S\) is additionally invariant under reversible processes, then \(W\le-\Delta F_\beta^A\) holds.
Consequently, even in GPT systems, no work can be extracted from an isothermal cyclic process composed of such processes, in agreement with the second law of thermodynamics.

For example, \(S_{\mathrm{meas}}\) is invariant under reversible processes and subadditive for the maximal tensor product, the canonical composition with a classical system, and the standard quantum tensor product.
Hence, whenever \(\mathcal M_0\otimes\mathrm{id}_{AB_1}\) does not decrease \(S_\mathrm{meas}\), the above argument applies.
Proposition \ref{prop:measurement entropy non-decresing measurement processes} ensures that \(\mathcal M_0\) does not decrease \(S_\mathrm{meas}\) if it is a repeatable fine-grained measurement process.
However, it remains unclear whether this entropy-nondecreasing property is stable under tensoring with \(\mathrm{id}_{AB_1}\), except in quantum theory or when \(AB_1\) is classical.

\section{Semipermeable-membrane model}

Several studies on semipermeable-membrane (SPM) models have investigated thermodynamics in GPTs \cite{Hanggi2013,Krumm2015,Krumm2017,Minagawa2022}.
In this section, we construct a general framework for SPM models and derive conditions under which measurement processes can be implemented by semipermeable membranes (SPMs) without violating the second law of thermodynamics.
Proofs of the statements in this section are provided in Appendix \ref{app:SPM model}.

\subsection{General Analysis}

We consider an ideal gas consisting of \(N\) classical free particles with an internal GPT degree of freedom \(A\), distributed among several containers.
The state of each particle can be described by a state \(\p{\rho_z^A}_{z\in Z}\) of the composite system of a classical system \(Z\), representing the container in which the particle is located, and the internal system \(A\).
Here, \(\V{\rho_z^A}\) represents the probability that the particle is in container \(z\in Z\), while \(\hat\rho_z^A\) denotes the corresponding internal state conditioned on the particle being in that container.

We consider SPMs whose transmission or reflection behavior depends on the internal state of the gas particles.
When a particle collides with such an SPM, a repeatable measurement process \(\mathcal M:A\to KA\) is performed.
If the measurement outcome is \(k\in K\), then the particle passes through the SPM with probability \(p_k\) and is reflected with probability \(1-p_k\).
By moving such SPMs within the containers, one can manipulate the positions of particles depending on their internal states.

Since the measurement process \(\mathcal M\) is repeatable, this operation is equivalent to first performing \(\mathcal M\) simultaneously on all particles in the container, thereby labeling each particle by its measurement outcome \(k \in K\), and then independently manipulating the particles corresponding to each label using SPMs that classically distinguish the labels.
The same interpretation applies when several SPMs implementing the same measurement process \(\mathcal M\) are used simultaneously, even if their transmission/reflection probabilities differ.
Furthermore, one may simultaneously insert, move, or remove walls that block all particles without measurements.

Let \(\mathcal F_k:Z\to Z'\) denote the process experienced by particles with measurement outcome \(k\in K\).
Then, the above operations implemented by SPMs can generally be represented as in Fig.~\ref{fig:SPM operation}.

\begin{figure}
    \centering
    \begin{tikzpicture}
        \draw[line width=1pt](-0.25,0)node[anchor=south west]{\(\scriptstyle\!Z\)}--(4.25,0);
        \draw[line width=1pt](-0.25,2.4)node[anchor=south west]{\(\scriptstyle\!A\)}--(4.25,2.4);
        \draw[line width=0.8pt,fill=black!10](0.5,0.7)rectangle(1.5,2.9);
        \draw(1,1.8)node{\(\mathcal M\)};
        \draw[line width=1pt](1.5,1.2)node[anchor=south west]{\(\scriptstyle\!k\in K\)}--(3.5,1.2);
        \draw(1.5,2.4)node[anchor=south west]{\(\scriptstyle\!A\)};
        \draw[line width=0.8pt,fill=black!10](2,-0.5)rectangle(3,0.5);
        \draw(2.5,0)node{\(\mathcal F_k\)};
        \draw[line width=1pt](2.5,0.5)--(2.5,1.2);
        \fill(2.5,1.2)circle[radius=3pt];
        \draw(3,0)node[anchor=south west]{\(\scriptstyle\!Z'\)};
        \draw[line width=1pt](3.5,0.95)--(3.5,1.45);
    \end{tikzpicture}
    \caption{General form of an operation on gas particles with internal GPT system \(A\), implemented using semipermeable membranes (SPMs) performing the measurement process \(\mathcal M\).}
    \label{fig:SPM operation}
\end{figure}
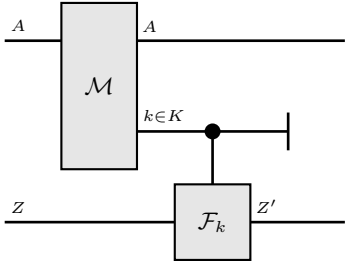

For example, the following operations are all included within this framework, assuming that the initial state is in equilibrium:
\begin{itemize}
    \item inserting an SPM into a container and waiting for equilibration,
    \item infinitesimally moving an SPM separating two containers,
    \item removing an SPM separating two containers and waiting for equilibration.
\end{itemize}
Furthermore, since the measurement process \(\mathcal M\) is repeatable, the identity
\begin{equation}
    \vcenter{\hbox{
    \begin{tikzpicture}[scale=0.5]
        \footnotesize
        \draw[line width=1pt](-0.25,0)node[anchor=south west]{\(\scriptstyle\!Z\)}--(7.75,0);
        \draw[line width=1pt](-0.25,2.4)node[anchor=south west]{\(\scriptstyle\!A\)}--(7.75,2.4);
        \draw[line width=0.8pt,fill=black!10](0.5,0.7)rectangle(1.5,2.9);
        \draw(1,1.8)node{\(\mathcal M\)};
        \draw[line width=1pt](1.5,1.2)node[anchor=south west]{\(\scriptstyle\!k\in K\)}--(3.5,1.2);
        \draw(1.5,2.4)node[anchor=south west]{\(\scriptstyle\!A\)};
        \draw[line width=0.8pt,fill=black!10](2,-0.5)rectangle(3,0.5);
        \draw(2.5,0)node{\(\mathcal F_k\)};
        \draw[line width=1pt](2.5,0.5)--(2.5,1.2);
        \fill(2.5,1.2)circle[radius=3pt];
        \draw(3,0)node[anchor=south west]{\(\scriptstyle\!Z'\)};
        \draw[line width=1pt](3.5,0.95)--(3.5,1.45);
        \draw[line width=0.8pt,fill=black!10](4,0.7)rectangle(5,2.9);
        \draw(4.5,1.8)node{\(\mathcal M\)};
        \draw[line width=1pt](5,1.2)node[anchor=south west]{\(\scriptstyle\!k'\in K\)}--(7,1.2);
        \draw(5,2.4)node[anchor=south west]{\(\scriptstyle\!A\)};
        \draw[line width=0.8pt,fill=black!10](5.5,-0.5)rectangle(6.5,0.5);
        \draw(6,0)node{\(\mathcal G_{k'}\)};
        \draw[line width=1pt](6,0.5)--(6,1.2);
        \fill(6,1.2)circle[radius=3pt];
        \draw(6.5,0)node[anchor=south west]{\(\scriptstyle\!Z''\)};
        \draw[line width=1pt](7,0.95)--(7,1.45);
    \end{tikzpicture}
    }}=\vcenter{\hbox{
    \begin{tikzpicture}[scale=0.5]
        \footnotesize
        \draw[line width=1pt](-0.25,0)node[anchor=south west]{\(\scriptstyle\!Z\)}--(4.75,0);
        \draw[line width=1pt](-0.25,2.4)node[anchor=south west]{\(\scriptstyle\!A\)}--(4.75,2.4);
        \draw[line width=0.8pt,fill=black!10](0.5,0.7)rectangle(1.5,2.9);
        \draw(1,1.8)node{\(\mathcal M\)};
        \draw[line width=1pt](1.5,1.2)node[anchor=south west]{\(\scriptstyle\!k\in K\)}--(4,1.2);
        \draw(1.5,2.4)node[anchor=south west]{\(\scriptstyle\!A\)};
        \draw[line width=0.8pt,fill=black!10](2,-0.5)rectangle(3.5,0.5);
        \draw(2.75,0)node{\(\mathcal G_k\mathcal F_k\)};
        \draw[line width=1pt](2.75,0.5)--(2.75,1.2);
        \fill(2.75,1.2)circle[radius=3pt];
        \draw(3.5,0)node[anchor=south west]{\(\scriptstyle\!Z''\)};
        \draw[line width=1pt](4,0.95)--(4,1.45);
    \end{tikzpicture}
    }}
\end{equation}
holds.
Therefore, any sequence of the above operations is also included in this framework.

We now consider the work extractable when such operations are performed isothermally.
Let the initial state be \(\rho_{\p0}^{ZA}\), the state immediately after the measurement process \(\mathcal M\) be \(\rho_{\p1}^{ZKA}\), and the state immediately after the process \(\mathcal F\) be \(\rho_{\p2}^{Z'KA}\).
Let \(V_z\) denote the volume of container \(z\in Z\) in the initial configuration, and \(V_{z'}\) the volume of container \(z'\in Z'\) in the final configuration.
Then, as is well known from the thermodynamics of ideal gases, the maximal work extractable during the process \(\mathcal F\) is given by
\begin{equation}
    \label{eq:SPM extractable work}
    \begin{aligned}
        W_\mathrm{ext}&=N\beta^{-1}\biggl(H^{Z'K}\p{\rho_{\p2}^{Z'K}}-H^{ZK}\p{\rho_{\p1}^{ZK}}\\
        &\qquad-\sum_{z\in Z}\V{\rho_{\p2z}}\log V_z+\sum_{z'\in Z'}\V{\rho_{\p0z'}}\log V_{z'}\biggr).
    \end{aligned}
\end{equation}

This work can actually be achieved using SPMs. To show this, we first introduce the protocols of ``separation'' and ``mixing'' (Fig.~\ref{fig:separation and mixing}).

\begin{figure}
    \centering
    \begin{tikzpicture}[scale=0.9]
        \fill[orange,opacity=0.05](-1,0.7)rectangle(1,-0.7);
        \fill[blue,opacity=0.05](-1,0.7)rectangle(1,-0.7);
        \draw[thick] (-1,0.7)--(1,0.7)--(1,-0.7)--(-1,-0.7)--(-1,0.7);
        \draw[dashed] (-0.95,0.7)--(-0.95,-0.7);
        \draw[dashed] (0.95,0.7)--(0.95,-0.7);
        \node at(0,0) {\(\rho\)};
        \draw[->, thick] (1.5,0)--(2,0);
        \fill[orange,opacity=0.075](2.5,0.7)rectangle(4,-0.7);
        \fill[blue,opacity=0.075](3,0.7)rectangle(4.5,-0.7);
        \draw[thick] (2.5,0.7)--(4.5,0.7)--(4.5,-0.7)--(2.5,-0.7)--(2.5,0.7);
        \draw[dashed] (3,0.7)--(3,-0.7);
        \draw[dashed] (4,0.7)--(4,-0.7);
        \draw[orange,<-] (2.8,0.3)--(3.2,0.3);
        \draw[orange,->] (3.8,0.45)--(3.95,0.3)--(3.8,0.15);
        \draw[blue,->] (3.8,-0.3)--(4.2,-0.3);
        \draw[blue,<-] (3.2,-0.45)--(3.05,-0.3)--(3.2,-0.15);
        \draw[->, thick] (5,0)--(5.5,0);
        \fill[orange,opacity=0.1](6,0.7)rectangle(7,-0.7);
        \fill[blue,opacity=0.1](7,0.7)rectangle(8,-0.7);
        \draw[thick] (6,0.7)--(8,0.7)--(8,-0.7)--(6,-0.7)--(6,0.7);
        \draw[dashed] (6.97,0.7)--(6.97,-0.7);
        \draw[dashed] (7.03,0.7)--(7.03,-0.7);
        \node[orange] at(6.5,0) {\(\mathcal{M}_0\rho\)};
        \node[blue] at(7.5,0) {\(\mathcal{M}_1\rho\)};
    \end{tikzpicture}\\
    \vspace{2ex}
    (a) separation\\
    \vspace{4ex}
    \begin{tikzpicture}[scale=0.9]
        \fill[orange,opacity=0.1](-1,0.7)rectangle(0,-0.7);
        \fill[blue,opacity=0.1](0,0.7)rectangle(1,-0.7);
         \draw[thick] (-1,0.7)--(1,0.7)--(1,-0.7)--(-1,-0.7)--(-1,0.7);
        \draw[dashed] (-0.03,0.7)--(-0.03,-0.7);
        \draw[dashed] (0.03,0.7)--(0.03,-0.7);
        \node[orange] at(-0.5,0) {\(\rho_0\)};
        \node[blue] at(0.5,0) {\(\rho_1\)};
        \draw[->, thick] (1.5,0)--(2,0);
        \fill[orange,opacity=0.075](2.5,0.7)rectangle(4,-0.7);
        \fill[blue,opacity=0.075](3,0.7)rectangle(4.5,-0.7);
        \draw[thick] (2.5,0.7)--(4.5,0.7)--(4.5,-0.7)--(2.5,-0.7)--(2.5,0.7);
        \draw[dashed] (3,0.7)--(3,-0.7);
        \draw[dashed] (4,0.7)--(4,-0.7);
        \draw[orange,->] (2.8,0.3)--(3.2,0.3);
        \draw[orange,->] (3.8,0.45)--(3.95,0.3)--(3.8,0.15);
        \draw[blue,<-] (3.8,-0.3)--(4.2,-0.3);
        \draw[blue,<-] (3.2,-0.45)--(3.05,-0.3)--(3.2,-0.15);
        \draw[->, thick] (5,0)--(5.5,0);
        \fill[orange,opacity=0.05](6,0.7)rectangle(8,-0.7);
        \fill[blue,opacity=0.05](6,0.7)rectangle(8,-0.7);
        \draw[thick] (6,0.7)--(8,0.7)--(8,-0.7)--(6,-0.7)--(6,0.7);
        \draw[dashed] (6.05,0.7)--(6.05,-0.7);
        \draw[dashed] (7.95,0.7)--(7.95,-0.7);
        \node at(7,0) {\(\rho_0+\rho_1\)};
    \end{tikzpicture}\\
    \vspace{2ex}
    (b) mixing
    \caption{Protocols of (a) separation and (b) mixing implemented using SPMs. The SPMs transmit or reflect gas particles depending on the outcome of the repeatable measurement process \(\mathcal M=\p{\mathcal M_0,\mathcal M_1}\).}
    \label{fig:separation and mixing}
\end{figure}
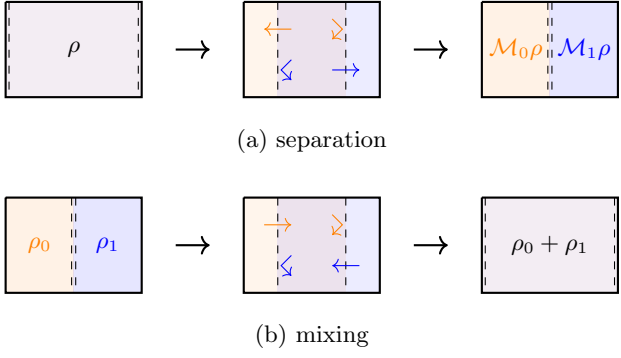

\begin{itemize}
    \item Separation:

    Consider a repeatable binary measurement \(\mathcal M=\p{\mathcal M_0,\mathcal M_1}\).
    Using an SPM that transmits only particles with outcome \(0\) and another SPM that transmits only particles with outcome \(1\), one can separate the gas into two gases with perfectly distinguishable internal states by the process illustrated in Fig.~\ref{fig:separation and mixing}(a).
    Measurement processes with more than two outcomes can similarly be treated by separating the components one by one.
    \item Mixing:

    Consider a repeatable binary measurement \(\mathcal M=\p{\mathcal M_0,\mathcal M_1}\).
    Using the same pair of SPMs, one can mix two gases with perfectly distinguishable internal states by the process illustrated in Fig.~\ref{fig:separation and mixing}(b), assuming that \(\mathcal M_k\rho_k=\rho_k\) for \(k=0,1\).
    Measurement processes with more than two outcomes can similarly be treated by successively mixing the components one by one.
\end{itemize}

A general process \(\mathcal F\) can be realized by combining these two protocols.
Let \(\mathcal F_{kzz'}\coloneq\epsilon_{z'}\mathcal F_k\delta_z\) denote the proportion of particles initially in container \(z\in Z\) with outcome \(k\in K\) that are transferred to container \(z'\in Z'\) by the process \(\mathcal F_k\).
Then the protocol is described as follows:
\begin{equation}
    \begin{aligned}
        &\p{\rho_{\p0z}^A}_{z\in Z}=\rho_{\p0}^{ZA}\\
        &\xmapsto{\ 1\ }\p{\mathcal M_k\rho_{\p0z}^A}_{z\in Z,k\in K}=\rho_{\p1}^{ZKA}\\
        &\xmapsto{\ 2\ }\p{\mathcal F_{kzz'}\mathcal M_k\rho_{\p0z}^A}_{z\in Z,z'\in Z',k\in K}\\
        &\xmapsto{\ 3\ }\p{\sum_{z\in Z}\mathcal F_{kzz'}\mathcal M_k\rho_{\p0z}^A}_{z'\in Z',k\in K}=\rho_{\p2}^{Z'KA}\\
        &\xmapsto{\ 4\ }\p{\sum_{k\in K}\sum_{z\in Z}\mathcal F_{kzz'}\mathcal M_k\rho_{\p0z}^A}_{z'\in Z'}=\rho_{\p2}^{Z'A}
    \end{aligned}
\end{equation}
The four steps are as follows:
\begin{enumerate}
    \item Perform ``separation'' independently in each container \(z\in Z\).
    \item Divide each container \(\p{z,k}\in ZK\) according to the proportions transferred to each \(z'\in Z'\).
    \item Merge the containers differing only in \(z\in Z\).
    \item Perform ``mixing'' independently for each \(z'\in Z'\).
\end{enumerate}
The work extracted during this protocol coincides with Eq.~\eqref{eq:SPM extractable work}.

We consider an isothermal cycle composed of operations implemented by SPMs as described above and refer to such a cycle as an \textit{SPM cycle}.
The general form of an SPM cycle is illustrated in Fig.~\ref{fig:SPM cycle}.
The cycle consists of \(n\) steps, and the \(i\)-th step (\(i=0,\dots,n-1\)) proceeds as follows.

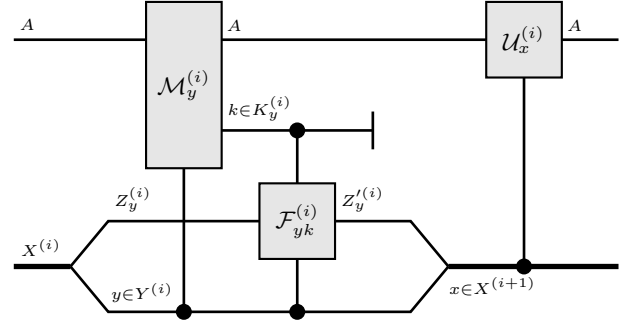
\begin{figure}
    \centering
    \begin{tikzpicture}
        \draw[line width=2pt](-0.25,0.6)node[anchor=south west]{\(\scriptstyle\!X^{\p i}\)}--(0.5,0.6);
        \draw[line width=1pt](-0.25,3.6)node[anchor=south west]{\(\scriptstyle\!A\)}--(7.75,3.6);
        \draw[line width=1pt](0.5,0.6)--(1,0)node[anchor=south west]{\(\scriptstyle\!\!y\in Y^{\p i}\)}--(5,0)--(5.5,0.6);
        \draw[line width=1pt](0.5,0.6)--(1,1.2)node[anchor=south west]{\(\scriptstyle\!Z^{\p i}_y\)}--(3,1.2);
        \draw[line width=1pt](2,0)--(2,1.9);
        \fill(2,0)circle[radius=3pt];
        \draw[line width=0.8pt,fill=black!10](1.5,1.9)rectangle(2.5,4.1);
        \draw(2,3)node{\(\mathcal M^{\p i}_y\)};
        \draw[line width=1pt](2.5,2.4)node[anchor=south west]{\(\scriptstyle\!k\in K^{\p i}_y\)}--(4.5,2.4);
        \draw(2.5,3.6)node[anchor=south west]{\(\scriptstyle\!A\)};
        \draw[line width=1pt](3.5,0)--(3.5,0.7);
        \draw[line width=1pt](3.5,1.7)--(3.5,2.4);
        \fill(3.5,0)circle[radius=3pt];
        \fill(3.5,2.4)circle[radius=3pt];
        \draw[line width=0.8pt,fill=black!10](3,0.7)rectangle(4,1.7);
        \draw(3.5,1.2)node{\(\mathcal F^{\p i}_{yk}\)};
        \draw[line width=1pt](4,1.2)node[anchor=south west]{\(\scriptstyle\!Z'^{\p i}_y\)}--(5,1.2)--(5.5,0.6);
        \draw[line width=1pt](4.5,2.15)--(4.5,2.65);
        \draw[line width=2pt](5.5,0.6)node[anchor=north west]{\(\scriptstyle\!\!\!x\in X^{\p{i+1}}\)}--(7.75,0.6);
        \draw[line width=0.8pt,fill=black!10](6,3.1)rectangle(7,4.1);
        \draw(6.5,3.6)node{\(\mathcal U^{\p i}_x\)};
        \draw[line width=1pt](6.5,0.6)--(6.5,3.1);
        \fill(6.5,0.6)circle[radius=3pt];
        \draw(7,3.6)node[anchor=south west]{\(\scriptstyle\!A\)};
    \end{tikzpicture}
    \caption{The \(i\)-th step in the general form of an SPM cycle, namely, an isothermal cycle implemented using SPMs.}
    \label{fig:SPM cycle}
\end{figure}

\begin{enumerate}
    \item The initial state is a state \(\rho_{\p{i,0}}\) of the composite system of a classical system \(X^{\p i}\) and an internal GPT system \(A\).
    \item A reversible process \(X^{\p i}\to\bigoplus_{y\in Y^{\p i}}Z^{\p i}_y\) is performed.
    \item Depending on the value \(y\in Y^{\p i}\), an operation implemented by SPMs is applied to the system \(Z^{\p i}_yA\):
    \begin{enumerate}
        \item[3-1.] A repeatable measurement process \(\mathcal M^{\p i}_y:A\to K^{\p i}_yA\) is performed.
        We denote the resulting state of the system \(\bigoplus_{y\in Y^{\p i}}Z^{\p i}_yK^{\p i}_yA\) by \(\rho_{\p{i,1}}\).
        \item[3-2.] A feedback process \(\mathcal F^{\p i}_{yk}:Z^{\p i}_y\to Z'^{\p i}_y\) depending on the value \(k\in K^{\p i}_y\) is performed.
        We denote the resulting state of the system \(\bigoplus_{y\in Y^{\p i}}Z'^{\p i}_yK^{\p i}_yA\) by \(\rho_{\p{i,2}}\).
        \item[3-3.] The system \(K^{\p i}_y\) is discarded.
    \end{enumerate}
    \item A reversible process \(\bigoplus_{y\in Y^{\p i}}Z'^{\p i}_y\to X^{\p{i+1}}\) is performed.
    \item A reversible process \(\mathcal U^{\p i}_x:A\to A\) depending on the value \(x\in X^{\p{i+1}}\) is performed.
    We denote the resulting state of the system \(X^{\p{i+1}}A\) by \(\rho_{\p{i+1,0}}\), which serves as the initial state for the next step.
\end{enumerate}

After \(n\) steps, the final state returns to the initial state: \(X^{\p n}=X^{\p0}\) and \(\rho_{\p{n,0}}=\rho_{\p{0,0}}\).

Here, the reversible process \(X^{\p i}\to\bigoplus_{y\in Y^{\p i}}Z^{\p i}_y\) represents relabeling the containers and partitioning them into groups indexed by \(y\in Y^{\p i}\).
Moreover, a process that inserts walls into a container and divides it into several containers is incorporated as a step in which the measurement process \(\mathcal M^{\p i}_y\) is trivial.
Previous works \cite{Hanggi2013,Krumm2015,Krumm2017,Minagawa2022} also allow reversible transformations of the internal states of the gas independently in each container; such operations are represented here by \(\mathcal{U}_{x}^{\p{i}}\).

We denote by
\begin{equation}
    W_\mathrm{ext}\coloneq\sum_{i=0}^{n-1}W_\mathrm{ext}^{\p i}
\end{equation}
the total extractable work of the SPM cycle, where each \(W_\mathrm{ext}^{\p i}\) is given by the analogue of Eq.~\eqref{eq:SPM extractable work} for the \(i\)-th step.

We then derive the following upper bound on the work extractable from the SPM cycle.
\begin{lem}
    \label{lem:SPM extractable work}
    Let \(S^A\) be a concave entropy that is invariant under reversible processes.
    Then
    \begin{equation}
        \begin{aligned}
            W_{\mathrm{ext}}&\le N\beta^{-1}\sum_{i=0}^{n-1}\bigggl(\eval{\aab{\aab{\tilde S^{K^{\p i}_yA}-S^A}_{Z'^{\p i}_y}}_{y\in Y^{\p i}}}_{\p{i,2}}\\
            &\hspace{8em}{}+\eval{\aab{\aab{S^A}_{Z^{\p i}_y}}_{y\in Y^{\p i}}}_{\p{i,0}}\hspace{-5pt}\\
            &\hspace{8em}{}-\eval{\aab{\aab{\tilde S^{K^{\p i}_yA}}_{Z^{\p i}_y}}_{y\in Y^{\p i}}}_{\p{i,1}}\bigggr).
        \end{aligned}
    \end{equation}
\end{lem}
Here, we use the abbreviation
\begin{equation}
    \begin{aligned}
        &\eval{\aab{\aab{S^A}_{Z'^{\p i}_y}}_{y\in Y^{\p i}}}_{\p{i,2}}\\
        &\quad\coloneq\aab{\aab{S^A}_{Z'^{\p i}_y}}_{y\in Y^{\p i}}\p{\p{\mathrm{tr}_{K_y^{\p i}}\rho_{\p{i,2}y}}_{y\in Y^{\p i}}}
    \end{aligned}
\end{equation}
and similarly for other terms.

\begin{rmk}
    In the setting considered here, where all but one of the systems are classical, concavity and subadditivity are equivalent for an entropy \(S\) that is consistent with the direct sum, as stated in Subsection \ref{Definitions and properties of entropies}.
    See also Remark \ref{rmk:Sagawa--Ueda}.
\end{rmk}

Consequently, we obtain the following conditions on measurement processes ensuring that no positive work can be extracted from an SPM cycle composed of such processes.
\begin{thm}
    \label{thm:weak SPM condition}
    For the SPM cycle, we have \(W_{\mathrm{ext}}\le 0\) if there exists a concave entropy that is invariant under reversible processes such that
    \begin{equation}
        \label{eq:weak SPM condition}
        \begin{aligned}
            S^A\p{\hat{\rho}_{\p{i,0}yz}^A}&\le\tilde S^{K^{\p i}_yA}\p{\mathcal M^{\p i}_y\hat{\rho}_{\p{i,0}yz}^A},\\
            \tilde S^{K^{\p i}_yA}\p{\mathcal M^{\p i}_y\hat{\rho}_{\p{i,2}yz'}^A}&\le S^A\p{\mathrm{tr}_{K^{\p i}_y}\mathcal M^{\p i}_y\hat{\rho}_{\p{i,2}yz'}^A}
        \end{aligned}
    \end{equation}
    for all \(i=0,\dots,n-1\), \(y\in Y^{\p i}\), \(z\in Z^{\p i}_y\) and \(z'\in Z'^{\p i}_y\).
\end{thm}

\begin{cor}
    \label{cor:strong SPM condition}
    Positive work cannot be extracted through an SPM cycle whenever there exists a concave entropy invariant under reversible processes such that every measurement process \(\mathcal M:A\to KA\) performed by SPMs satisfies
    \begin{equation}
        \label{eq:strong SPM condition}
        S^A\p{\rho^A}\le\tilde S^{KA}\p{\mathcal M\rho^A}\le S^A\p{\mathrm{tr}_K\mathcal M\rho^A}
    \end{equation}
    for all \(\rho^A\in\mathrm{St}\p A\).
\end{cor}

Theorem \ref{thm:weak SPM condition} and Corollary \ref{cor:strong SPM condition} have several important consequences.
First, when \(A\) is a quantum system and \(S^A\) is taken to be the von Neumann entropy, no positive work can be extracted from an SPM cycle using L\"uders measurement processes.

More generally, suppose that there exists an entropy \(S^A\) that is concave and invariant under reversible processes, and satisfies
\begin{equation}
    \label{Krumm's condition for entropy}
    S^A\p{\mathrm{tr}_X\rho^{XA}}=\tilde S^{XA}\p{\rho^{XA}}
\end{equation}
for every ensemble of perfectly distinguishable states \(\rho^{XA}=\p{\rho_x^A}_{x\in X}\).
Then positive work cannot be extracted whenever every measurement process \(\mathcal M^{\p i}_y\) satisfies
\begin{equation}
    S^A\p{\hat\rho_{\p{i,0}yz}^A}\le S^A\p{\mathrm{tr}_{K^{\p i}_y}\mathcal M^{\p i}_y\hat\rho_{\p{i,0}yz}^A}
\end{equation}
for all \(z\in Z^{\p i}_y\).
In particular, this condition is satisfied whenever every \(\mathcal M^{\p i}_y\) leaves the corresponding input state undisturbed, namely \(\hat\rho_{\p{i,0}yz}^A=\mathrm{tr}_{K^{\p i}_y}\mathcal M^{\p i}_y\hat\rho_{\p{i,0}yz}^A\).

The measurement entropy \(S_\mathrm{meas}\) is concave, invariant under reversible processes, and consistent with the direct sum.
Likewise, the infinity entropy \(S_\mathrm{acc}^\infty\) associated with the accessible information entropy is invariant under reversible processes and consistent with the direct sum, and is concave provided that the conjectured Holevo bound in Ref.~\cite{Kimura2016} holds.
Thus, these entropies satisfy the assumptions of Theorem \ref{thm:weak SPM condition} and Corollary \ref{cor:strong SPM condition}.

On the other hand, the accessible information entropy \(S_\mathrm{acc}\) is invariant under reversible processes and consistent with the direct sum but is not necessarily concave \cite{Kimura2010}.
Moreover, \(S_\mathrm{acc}\) satisfies the second inequality of Eq.~\eqref{eq:strong SPM condition} for strongly repeatable measurement processes.
Therefore, if every measurement process \(\mathcal M\) performed by SPMs is strongly repeatable and satisfies \(S_\mathrm{acc}^A\p{\rho^A}\le S_\mathrm{acc}^{KA}\p{\mathcal M\rho^A}\) but nevertheless \(W_\mathrm{ext}>0\), then the positive work extraction must originate from the failure of concavity of \(S_\mathrm{acc}\).

Finally, the comparison between the measurement entropy \(S_\mathrm{meas}\) and the accessible information entropy \(S_\mathrm{acc}\) gives further insight into the origin of possible work extraction.
When every \(\mathcal M^{\p i}_y\) is a strongly repeatable fine-grained measurement process, the extractable work is bounded as
\begin{equation}
    \begin{aligned}
        W_{\mathrm{ext}}&\le N\beta^{-1}\sum_{i=0}^{n-1}\:\biggl\langle\Bigl\langle S_{\mathrm{acc}}^A-S_{\mathrm{meas}}^A\\
        &\qquad-\aab{S_{\mathrm{acc}}^A-S_{\mathrm{meas}}^A}_{K^{\p i}_y}\Bigr\rangle_{Z'^{\p i}_y}\biggr\rangle_{y\in Y^{\p i}}\Bigg|_{\p{i,2}}.
    \end{aligned}
\end{equation}
Thus, any positive work extraction from such an SPM cycle must originate from the discrepancy between \(S_\mathrm{meas}^A\) and \(S_\mathrm{acc}^A\).

These observations highlight a characteristic feature of GPTs: unlike in quantum theory, there is generally no unique entropy satisfying all the fundamental properties of the von Neumann entropy.
As a result, the relation between entropy nondecrease under measurement processes and the impossibility of work extraction becomes nontrivial in GPTs.
In particular, different entropic properties, such as concavity, consistency with the direct sum, and monotonicity under measurement processes, play distinct roles in constraining thermodynamic cycles.
The above results therefore provide a structural understanding of how thermodynamic consistency in GPTs depends on the interplay among different entropic properties.

\subsection{Examples}
\label{Examples}

In this subsection, we present examples of SPM cycles from which positive work can be extracted.

\subsubsection{Square system}

We consider an ideal gas whose internal GPT state space \(A\) is given by the square
\begin{equation}
    \mathrm{St}\p A\coloneq\mathrm{conv}\B{\rho_n\coloneq\p{\cos\frac{n\pi}2,\sin\frac{n\pi}2,1}\in\mathds R^3\mid n\in\mathds Z}.
\end{equation}
The square system is also known as the gbit and appears as a local subsystem in GPT realizations of the Popescu--Rohrlich box \cite{Popescu1994,Barrett2007}.

We consider the SPM cycle illustrated in Fig.~\ref{fig:square}.
Initially, there are four containers, each containing an ideal gas of \(N/4\) particles with internal state \(\frac34\rho_n+\frac14\rho_{n+2}\) for \(n=0,\dots,3\).

\begin{figure}
    \begin{tikzpicture}[scale=1.25]
        \def\scale{0.5}
        \draw[left](-0.75,-0.75)node{(i)};
        \draw[left](-0.75,-3)node{(ii)};
        \draw(-0.5,-3)node{\(\begin{cases}\\\\\\\\\\\\\end{cases}\)};
        \draw[left](-0.75,-5.25)node{(iii)};
        \draw[left](-0.75,-6.75)node{(iv)};
        \draw[left](-0.75,-8.25)node{(v)};
        \def\square{
            \coordinate(a0)at({cos((\x+0)*90)},{sin((\x+0)*90)});
            \coordinate(a1)at({cos((\x+1)*90)},{sin((\x+1)*90)});
            \coordinate(a2)at({cos((\x+2)*90)},{sin((\x+2)*90)});
            \coordinate(a3)at({cos((\x+3)*90)},{sin((\x+3)*90)});
            \draw(a0)--(a1)--(a2)--(a3)--cycle;
        }
        \foreach\x in{0,...,3}{
            \draw(\x,0)node{
                \begin{tikzpicture}[scale=\scale]
                    \square
                    \draw[dotted](a0)--(a2);
                    \draw[dotted]({cos(\x*90)/2+cos((\x+1)*90)/2},{sin(\x*90)/2+sin((\x+1)*90)/2})--({cos(\x*90)/2+cos((\x-1)*90)/2},{sin(\x*90)/2+sin((\x-1)*90)/2});
                    \fill({cos(\x*90)/2},{sin(\x*90)/2})circle[radius=3pt];
                    \draw(0,0)node{\(\scriptstyle\frac N4\)};
                \end{tikzpicture}
            };
            \draw[->](\x,-0.5)--(\x,-1);
            \draw[left](\x,-0.75)node{s.};
            \draw(\x,-1.5)node{
                \begin{tikzpicture}[scale=\scale]
                    \square
                    \draw[dotted](a0)--(a2);
                    \draw[dotted]({cos(\x*90)/2+cos((\x+1)*90)/2},{sin(\x*90)/2+sin((\x+1)*90)/2})--({cos(\x*90)/2+cos((\x-1)*90)/2},{sin(\x*90)/2+sin((\x-1)*90)/2});
                    \fill(a0)circle[radius=3pt];
                    \fill(a2)circle[radius=3pt];
                    \draw({cos(\x*90)/2.5},{sin(\x*90)/2.5})node{\(\scriptstyle\frac{3N}{16}\)};
                    \draw(-{cos(\x*90)/2},-{sin(\x*90)/2})node{\(\scriptstyle\frac N{16}\)};
                \end{tikzpicture}
            };
        };
        \draw[->](-0.1,-2)--(-0.1,-2.5);
        \draw[->](0.1,-2)--(1.9,-2.5);
        \draw[->](0.9,-2)--(0.9,-2.5);
        \draw[->](1.1,-2)--(2.9,-2.5);
        \draw[->](2.1,-2)--(2.1,-2.5);
        \draw[->](1.9,-2)--(0.1,-2.5);
        \draw[->](3.1,-2)--(3.1,-2.5);
        \draw[->](2.9,-2)--(1.1,-2.5);
        \foreach\x in{0,...,3}\draw(\x,-3)node{
            \begin{tikzpicture}[scale=\scale]
                \square
                \fill(a0)circle[radius=3pt];
                \draw({cos(\x*90)/2},{sin(\x*90)/2})node{\(\scriptstyle\frac N4\)};
            \end{tikzpicture}
        };
        \draw[->](-0.1,-3.5)--(-0.1,-4);
        \draw[->](0.1,-3.5)--(2.9,-4);
        \draw[->](1.1,-3.5)--(0.9,-4);
        \draw[->](0.9,-3.5)--(0.1,-4);
        \draw[->](2.1,-3.5)--(1.9,-4);
        \draw[->](1.9,-3.5)--(1.1,-4);
        \draw[->](3.1,-3.5)--(3.1,-4);
        \draw[->](2.9,-3.5)--(2.1,-4);
        \foreach\x in{0,...,3}{
            \draw(\x,-4.5)node{
                \begin{tikzpicture}[scale=\scale]
                    \square
                    \fill(a0)circle[radius=3pt];
                    \fill(a1)circle[radius=3pt];
                    \draw({cos(\x*90)/2},{sin(\x*90)/2})node{\(\scriptstyle\frac N8\)};
                    \draw({cos((\x+1)*90)/2},{sin((\x+1)*90)/2})node{\(\scriptstyle\frac N8\)};
                \end{tikzpicture}
            };
            \draw[->](\x,-5)--(\x,-5.5);
            \draw[left](\x,-5.25)node{m.};
            \draw(\x,-6)node{
                \begin{tikzpicture}[scale=\scale]
                    \square
                    \fill({cos(\x*90)/2+cos((\x+1)*90)/2},{sin(\x*90)/2+sin((\x+1)*90)/2})circle[radius=3pt];
                    \draw({cos(\x*90)/5+cos((\x+1)*90)/5},{sin(\x*90)/5+sin((\x+1)*90)/5})node{\(\scriptstyle\frac N4\)};
                \end{tikzpicture}
            };
        };
        \draw[->](-0.1,-6.5)--(-0.1,-7);
        \draw[->](0.1,-6.5)--(0.9,-7);
        \draw[->](0.9,-6.5)--(1.1,-7);
        \draw[->](1.1,-6.5)--(1.9,-7);
        \draw[->](1.9,-6.5)--(2.1,-7);
        \draw[->](2.1,-6.5)--(2.9,-7);
        \draw[->](3.1,-6.5)--(3.1,-7);
        \draw[->](2.9,-6.5)--(0.1,-7);
        \foreach\x in{0,...,3}{
            \draw(\x,-7.5)node{
                \begin{tikzpicture}[scale=\scale]
                    \square
                    \draw[dotted](a0)--(a2);
                    \draw[dotted]({cos(\x*90)/2+cos((\x+1)*90)/2},{sin(\x*90)/2+sin((\x+1)*90)/2})--({cos(\x*90)/2+cos((\x-1)*90)/2},{sin(\x*90)/2+sin((\x-1)*90)/2});
                    \fill({cos(\x*90)/2+cos((\x+1)*90)/2},{sin(\x*90)/2+sin((\x+1)*90)/2})circle[radius=3pt];
                    \fill({cos(\x*90)/2+cos((\x-1)*90)/2},{sin(\x*90)/2+sin((\x-1)*90)/2})circle[radius=3pt];
                    \draw({cos(\x*90)/15+cos((\x+1)*90)/2.5},{sin(\x*90)/15+sin((\x+1)*90)/2.5})node{\(\scriptstyle\frac N8\)};
                    \draw({cos(\x*90)/15+cos((\x-1)*90)/2.5},{sin(\x*90)/15+sin((\x-1)*90)/2.5})node{\(\scriptstyle\frac N8\)};
                \end{tikzpicture}
            };
            \draw[->](\x,-8)--(\x,-8.5);
            \draw(\x,-9)node{
                \begin{tikzpicture}[scale=\scale]
                    \square
                    \draw[dotted](a0)--(a2);
                    \draw[dotted]({cos(\x*90)/2+cos((\x+1)*90)/2},{sin(\x*90)/2+sin((\x+1)*90)/2})--({cos(\x*90)/2+cos((\x-1)*90)/2},{sin(\x*90)/2+sin((\x-1)*90)/2});
                    \fill({cos(\x*90)/2},{sin(\x*90)/2})circle[radius=3pt];
                    \draw(0,0)node{\(\scriptstyle\frac N4\)};
                \end{tikzpicture}
            };
        };
    \end{tikzpicture}
    \caption{An SPM cycle from which positive work \(W_\mathrm{ext}=\frac N4\beta^{-1}\log\frac{27}{16}\) can be extracted for a gas whose internal GPT state space is a square. In the figure, ``s.'' and ``m.'' denote the separation and mixing protocols (Fig.~\ref{fig:separation and mixing}), respectively.}
    \label{fig:square}
\end{figure}
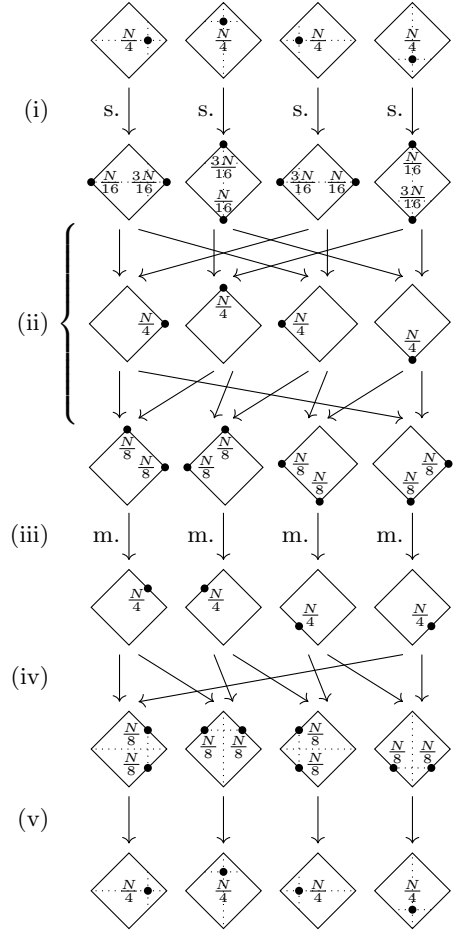

We first perform the separation protocol (Fig.~\ref{fig:separation and mixing}(a)) on each gas using SPMs implementing measure-and-prepare processes for the perfectly distinguishable pairs \(\p{\rho_n,\rho_{n+2}}\) for \(n=0,\dots,3\).
After this step, each original container is separated into two containers: one containing \(3N/16\) particles in state \(\rho_n\), and the other containing \(N/16\) particles in state \(\rho_{n+2}\) (Fig.~\ref{fig:square}(i)).
The work required for this separation step is
\begin{equation}
    N\beta^{-1}H\p{\p{\frac34,\frac14}}.
\end{equation}

Next, containers containing gases with the same internal state are combined and repartitioned, yielding four pairs of containers, each containing \(N/8\) particles in state \(\rho_n\) and \(N/8\) particles in state \(\rho_{n+1}\) for \(n=0,\dots,3\) (Fig.~\ref{fig:square}(ii)).

We then perform the mixing protocol (Fig.~\ref{fig:separation and mixing}(b)) on each pair using SPMs implementing measure-and-prepare processes for the perfectly distinguishable pairs \(\p{\rho_n,\rho_{n+1}}\) for \(n=0,\dots,3\).
As a result, we obtain four containers, each containing \(N/4\) particles in state \(\frac12\rho_n+\frac12\rho_{n+1}\) (Fig.~\ref{fig:square}(iii)).
The work extracted in this mixing step is
\begin{equation}
    N\beta^{-1}H\p{\p{\frac12,\frac12}}.
\end{equation}

Next, the containers are repartitioned into four pairs, each containing \(N/8\) particles in state \(\frac12\rho_n+\frac12\rho_{n+1}\) and \(N/8\) particles in state \(\frac12\rho_n+\frac12\rho_{n-1}\) for \(n=0,\dots,3\) (Fig.~\ref{fig:square}(iv)).

Finally, we mix each pair without measurement.
After this step, the system returns to the initial state (Fig.~\ref{fig:square}(v)).

Therefore, the total extractable work of the cycle is
\begin{equation}
    \begin{aligned}
        W_\mathrm{ext}&=N\beta^{-1}H\p{\p{\frac12,\frac12}}-N\beta^{-1}H\p{\p{\frac34,\frac14}}\\
        &=\frac N4\beta^{-1}\log\frac{27}{16}>0.
    \end{aligned}
\end{equation}
Thus, positive work can be extracted from the SPM cycle of the square system.
In particular, the existence of SPMs implementing the above measurement processes without work cost would contradict the second law of thermodynamics.

The measurement processes used here do not disturb the input states.
Moreover, since they are strongly repeatable and do not decrease \(S_\mathrm{acc}^A\) of the input states, the positive work extraction in this cycle reflects the failure of concavity of \(S_\mathrm{acc}^A\) in the square system.
Furthermore, since the measurement processes are also fine-grained, the work extraction can equivalently be understood as arising from the discrepancy between \(S_\mathrm{meas}^A\) and \(S_\mathrm{acc}^A\) in the square system.
See Ref.~\cite{Kimura2010} for explicit forms of \(S_{\mathrm{meas}}\) and \(S_{\mathrm{acc}}\) in the square model.

\subsubsection{Regular hexagon system}

We next consider the regular hexagon system \cite{Janotta2011},
\begin{equation}
    \mathrm{St}\p A\coloneq\mathrm{conv}\B{\rho_n\coloneq\p{\cos\frac{n\pi}3,\sin\frac{n\pi}3,1}\in\mathds R^3\mid n\in\mathds Z}.
\end{equation}

We consider the SPM cycle illustrated in Fig.~\ref{fig:hexagon}.
Initially, there are six containers, each containing an ideal gas of \(N/6\) particles with internal state \(\frac12\rho_{n-1}+\frac12\rho_{n+1}\) for \(n=0,\dots,5\).

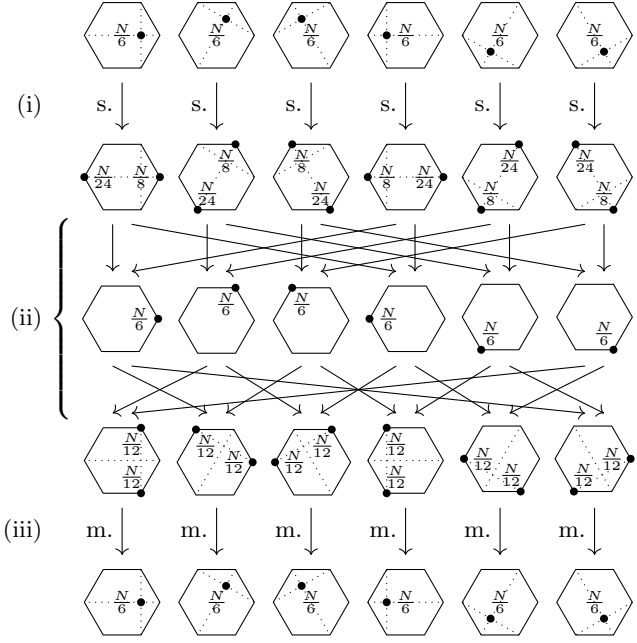
\begin{figure}
    \begin{tikzpicture}[scale=1.25]
        \def\scale{0.5}
        \draw[left](-0.75,-0.75)node{(i)};
        \draw[left](-0.75,-3)node{(ii)};
        \draw(-0.5,-3)node{\(\begin{cases}\\\\\\\\\\\\\end{cases}\)};
        \draw[left](-0.75,-5.25)node{(iii)};
        \def\hexagon{
            \coordinate(a0)at({cos((\x+0)*60)},{sin((\x+0)*60)});
            \coordinate(a1)at({cos((\x+1)*60)},{sin((\x+1)*60)});
            \coordinate(a2)at({cos((\x+2)*60)},{sin((\x+2)*60)});
            \coordinate(a3)at({cos((\x+3)*60)},{sin((\x+3)*60)});
            \coordinate(a4)at({cos((\x+4)*60)},{sin((\x+4)*60)});
            \coordinate(a5)at({cos((\x+5)*60)},{sin((\x+5)*60)});
            \draw(a0)--(a1)--(a2)--(a3)--(a4)--(a5)--cycle;
        }
        \foreach\x in{0,...,5}{
            \draw(\x,0)node{
                \begin{tikzpicture}[scale=\scale]
                    \hexagon
                    \draw[dotted](a0)--(a3);
                    \draw[dotted](a1)--(a5);
                    \fill({cos(\x*60)/2},{sin(\x*60)/2})circle[radius=3pt];
                    \draw(0,0)node{\(\scriptstyle\frac N6\)};
                \end{tikzpicture}
            };
            \draw[->](\x,-0.5)--(\x,-1);
            \draw[left](\x,-0.75)node{s.};
            \draw(\x,-1.5)node{
                \begin{tikzpicture}[scale=\scale]
                    \hexagon
                    \draw[dotted](a0)--(a3);
                    \draw[dotted](a1)--(a5);
                    \fill(a0)circle[radius=3pt];
                    \fill(a3)circle[radius=3pt];
                    \draw({cos(\x*60)/2},{sin(\x*60)/2})node{\(\scriptstyle\frac N8\)};
                    \draw(-{cos(\x*60)/2},-{sin(\x*60)/2})node{\(\scriptstyle\frac N{24}\)};
                \end{tikzpicture}
            };
        };
        \draw[->](-0.1,-2)--(-0.1,-2.5);
        \draw[->](0.1,-2)--(2.9,-2.5);
        \draw[->](0.9,-2)--(0.9,-2.5);
        \draw[->](1.1,-2)--(3.9,-2.5);
        \draw[->](1.9,-2)--(1.9,-2.5);
        \draw[->](2.1,-2)--(4.9,-2.5);
        \draw[->](3.1,-2)--(3.1,-2.5);
        \draw[->](2.9,-2)--(0.1,-2.5);
        \draw[->](4.1,-2)--(4.1,-2.5);
        \draw[->](3.9,-2)--(1.1,-2.5);
        \draw[->](5.1,-2)--(5.1,-2.5);
        \draw[->](4.9,-2)--(2.1,-2.5);
        \foreach\x in{0,...,5}\draw(\x,-3)node{
            \begin{tikzpicture}[scale=\scale]
                \hexagon
                \fill(a0)circle[radius=3pt];
                \draw({cos(\x*60)/2},{sin(\x*60)/2})node{\(\scriptstyle\frac N6\)};
            \end{tikzpicture}
        };
        \draw[->](-0.1,-3.5)--(0.9,-4);
        \draw[->](0.1,-3.5)--(4.9,-4);
        \draw[->](1.1,-3.5)--(1.9,-4);
        \draw[->](0.9,-3.5)--(-0.1,-4);
        \draw[->](2.1,-3.5)--(2.9,-4);
        \draw[->](1.9,-3.5)--(1.1,-4);
        \draw[->](3.1,-3.5)--(3.9,-4);
        \draw[->](2.9,-3.5)--(2.1,-4);
        \draw[->](4.1,-3.5)--(5.1,-4);
        \draw[->](3.9,-3.5)--(3.1,-4);
        \draw[->](4.9,-3.5)--(0.1,-4);
        \draw[->](5.1,-3.5)--(4.1,-4);
        \foreach\x in{0,...,5}{
            \draw(\x,-4.5)node{
                \begin{tikzpicture}[scale=\scale]
                    \hexagon
                    \draw[dotted](a0)--(a3);
                    \draw[dotted](a1)--(a5);
                    \fill(a1)circle[radius=3pt];
                    \fill(a5)circle[radius=3pt];
                    \draw({cos((\x+1)*60)/2},{sin((\x+1)*60)/2})node{\(\scriptstyle\frac N{12}\)};
                    \draw({cos((\x-1)*60)/2},{sin((\x-1)*60)/2})node{\(\scriptstyle\frac N{12}\)};
                \end{tikzpicture}
            };
            \draw[->](\x,-5)--(\x,-5.5);
            \draw[left](\x,-5.25)node{m.};
            \draw(\x,-6)node{
                \begin{tikzpicture}[scale=\scale]
                    \hexagon
                    \draw[dotted](a0)--(a3);
                    \draw[dotted](a1)--(a5);
                    \fill({cos(\x*60)/2},{sin(\x*60)/2})circle[radius=3pt];
                    \draw(0,0)node{\(\scriptstyle\frac N6\)};
                \end{tikzpicture}
            };
        };
    \end{tikzpicture}
    \caption{An SPM cycle from which positive work \(W_\mathrm{ext}=\frac N4\beta^{-1}\log\frac{27}{16}\) can be extracted for a gas whose internal GPT state space is a regular hexagon. In the figure, ``s.'' and ``m.'' denote the separation and mixing protocols (Fig.~\ref{fig:separation and mixing}), respectively.}
    \label{fig:hexagon}
\end{figure}

We first perform the separation protocol (Fig.~\ref{fig:separation and mixing}(a)) on each gas using SPMs implementing measure-and-prepare processes for the perfectly distinguishable pairs \(\p{\rho_n,\rho_{n+3}}\) for \(n=0,\dots,5\).
After this step, each original container is separated into two containers: one containing \(N/8\) particles in state \(\rho_n\), and the other containing \(N/24\) particles in state \(\rho_{n+3}\) (Fig.~\ref{fig:hexagon}(i)).
The work required for this separation step is
\begin{equation}
    N\beta^{-1}H\p{\p{\frac34,\frac14}}.
\end{equation}

Next, containers containing gases with the same internal state are combined and repartitioned, yielding six pairs of containers, each containing \(N/12\) particles in state \(\rho_{n-1}\) and \(N/12\) particles in state \(\rho_{n+1}\) for \(n=0,\dots,5\) (Fig.~\ref{fig:hexagon}(ii)).

Finally, we perform the mixing protocol (Fig.~\ref{fig:separation and mixing}(b)) on each pair using SPMs implementing measure-and-prepare processes for the perfectly distinguishable pairs \(\p{\rho_{n-1},\rho_{n+1}}\) for \(n=0,\dots,5\).
After this step, the system returns to the initial state (Fig.~\ref{fig:hexagon}(iii)).
The work extracted in this mixing step is
\begin{equation}
    N\beta^{-1}H\p{\p{\frac12,\frac12}}.
\end{equation}

Therefore, the total extractable work of the cycle is again
\begin{equation}
    W_\mathrm{ext}=\frac N4\beta^{-1}\log\frac{27}{16}>0.
\end{equation}
Thus, we again obtain positive work extraction from the SPM cycle, implying that the existence of SPMs implementing the above measurement processes without work cost would contradict the second law of thermodynamics.

As in the square system, the measurement processes used here do not disturb the input states.
Moreover, since they are strongly repeatable fine-grained measurement processes, the positive work extraction in this cycle reflects the discrepancy between \(S_\mathrm{meas}^A\) and \(S_\mathrm{acc}^A\) in the regular hexagon system.

This cycle also provides a concrete realization of the result of Minagawa \textit{et al.}\ \cite{Minagawa2022}, which states that positive work can be extracted through an SPM cycle in systems containing states whose entropy, defined through decompositions into perfectly distinguishable pure states, is not uniquely determined.

\subsection{Relation to previous works}

In this section, we discuss the relation to previous works on SPM models \cite{Neumann1927,Neumann1932,Petz2001,Hanggi2013,Barnum2015,Krumm2015,Krumm2017,Takakura2019,Minagawa2022}.

\subsubsection{Von Neumann \cite{Neumann1927,Neumann1932,Petz2001}}

Research on SPM models originates from the thought experiments of von Neumann \cite{Neumann1927,Neumann1932}.
He derived the von Neumann entropy as the thermodynamic entropy of quantum systems by considering the following process.
Consider an ideal gas consisting of \(N\lambda_x\) particles whose internal quantum states are given by mutually orthogonal pure states \(\ket{\varphi_x}\) (\(x\in X\), \(\sum_{x\in X}\lambda_x=1\)).
He considered a process in which the gases are first separated at constant volume by SPMs and then compressed so that the total volume returns to its original value.
Since the first process is merely a parallel translation for each particle, the required work is zero.
Assuming further that the final pure states have zero entropy, one obtains that the entropy of the state \(\rho=\sum_{x\in X}\lambda_x\ket{\varphi_x}\bra{\varphi_x}\) is given by \(-\sum_{x\in X}\lambda_x\log\lambda_x\).

Furthermore, assuming thermodynamics, he argued that there cannot exist an SPM that separates two nonorthogonal pure states \(\ket\varphi\) and \(\ket\psi\) (note that such a process does not correspond to a measurement process).
Consider an ideal gas consisting of \(N/2\) particles in the internal states \(\ket\varphi\) and \(\ket\psi\), respectively.
Suppose that there exists an SPM that transmits only one of the two states.
By quasistatically moving such an SPM from the edge of the container to the center, one obtains a process in which the required mechanical work must coincide with the heat calculated from the change of the von Neumann entropy.
From this consistency condition, one concludes that \(\ket\varphi\) and \(\ket\psi\) must be orthogonal.

Von Neumann justified that separations using SPMs may be regarded as thermodynamically costless as follows \cite{Neumann1932}.
First, if the state is one of the states \(\ket{\varphi_x}\), the measurement implemented by the SPM does not disturb the state and therefore presumably leaves no trace on the measuring apparatus (the SPM).
Second, unlike Maxwell's demon, the SPM uses for feedback only the information about the internal state of the particle, and does not distinguish from which side of the SPM the particle came.
A related interpretation was later proposed by Maruyama \textit{et al.}\ \cite{Maruyama2005}, who argued that the information about the measurement outcomes in the SPM process can be erased without thermodynamic cost because it remains perfectly correlated with the internal states of the particles, unlike in the Szilard engine \cite{Szilard1929}.

The above argument of von Neumann derives the von Neumann entropy under the assumption that SPMs performing projective measurement processes without thermodynamic cost exist, provided that the input state is not disturbed.
By contrast, in our work, we prove the converse statement: the existence of the von Neumann entropy guarantees that even if SPMs performing L\"uders measurement processes for free exist, they do not contradict the second law of thermodynamics.
Here, we do not assume that the input state is undisturbed.

\subsubsection{H\"anggi and Wehner \cite{Hanggi2013}}

H\"anggi and Wehner \cite{Hanggi2013} considered the following cycle.
First, they considered two gases whose internal states are distinct mixed states, and mixed them using SPMs implementing a rank-\(1\) L\"uders measurement process, until the internal state became uniform.
Next, they used another SPM to decompose the gases into perfectly distinguishable pure (PDP) states.
Finally, they transformed these states into different pure states by reversible processes and mixed them again to recover the original two states.

Unlike other previous works \cite{Neumann1927,Neumann1932,Petz2001,Barnum2015,Krumm2015,Krumm2017,Minagawa2022}, which consider only measurement processes that do not disturb the input states, no such restriction is imposed here.

They showed that, in quantum systems, the work extractable from the above cycle is bounded by a quantity characterizing the uncertainty of the probability distributions of the measurement outcomes.
From this, they concluded that if the condition known as the fine-grained uncertainty relations (FGUR) \cite{Oppenheim2010}, which is satisfied in quantum theory, were violated, then positive work could be extracted.
However, as they themselves pointed out, one should note that even within the GPT framework, it is impossible to modify the theory in such a way that only the FGUR is violated while all other aspects of quantum theory remain unchanged.

They also extended the above bound to a restricted class of GPT systems satisfying pure transitivity, PDP decomposability, and a certain form of self-duality, and showed that the work extractable from the cycle is bounded by an uncertainty term.
However, unlike in the quantum case, no explicit bound on the uncertainty term was obtained in GPTs, and therefore it remained unclear whether positive work extraction is actually possible.

They also mentioned the possibility that discrepancies between different definitions of entropy in GPTs could lead to violations of the second law of thermodynamics, although no detailed analysis of this issue was given.

Our results provide a bound on the work extractable from the H\"anggi--Wehner cycle as well.
Indeed, the measurement processes appearing in their cycle are strongly repeatable fine-grained measurement processes.
Therefore, in systems where the measurement entropy \(S_\mathrm{meas}\) and the accessible information entropy \(S_\mathrm{acc}\) coincide, no positive work can be extracted from the cycle.

\subsubsection{Krumm \textit{et al.}\ \cite{Barnum2015,Krumm2015,Krumm2017,Takakura2019}}

By generalizing von Neumann's thought experiment, Krumm \textit{et al.}\ \cite{Krumm2015,Krumm2017} derived conditions that the ``thermodynamic entropy'' of an internal GPT system must satisfy.
Following von Neumann's argument, they assumeed that the decomposition into PDP states can be implemented by an SPM without thermodynamic cost.
Under this assumption, they showed that the condition in Eq.~\eqref{Krumm's condition for entropy}, namely,
\begin{equation}
    S^A\p{\mathrm{tr}_X\rho^{XA}}=\tilde S^{XA}\p{\rho^{XA}}
\end{equation}
must hold for any ensemble of perfectly distinguishable states \(\rho^{XA}=\p{\rho_x^A}_{x\in X}\).
In particular, under the assumption of pure transitivity and PDP decomposability, they derived the following expression for the thermodynamic entropy:
\begin{equation}
\label{PDP decomposition entropy}
    S^A\p{\rho^A}=H^X\p{\mathrm{tr}_A\rho^{XA}},
\end{equation}
where \(\rho^{XA}=\p{\rho_x^A}_{x\in X}\in\mathrm{Ens}_\mathrm{pure}\p{\rho^A}\) is an ensemble of PDP states.

They further showed that no entropy satisfying Eq.~\eqref{Krumm's condition for entropy} exists for the square system.
Later, Takakura \cite{Takakura2019} extended this result to all non-classical regular polygon systems.

Krumm \textit{et al.}\ \cite{Krumm2015,Krumm2017} also argued that the thermodynamic entropy must be concave by considering the free mixing process.

Furthermore, assuming PDP decomposability and a stronger version of pure transitivity, they showed that the entropy defined by Eq.~\eqref{PDP decomposition entropy} is well defined, satisfies the above conditions, coincides with the measurement entropy \(S_\mathrm{meas}\), and is bounded below by the mixing entropy \(S_\mathrm{mix}\).
They also introduced a natural generalization of L\"uders measurement processes in such systems and showed that such measurements do not decrease the entropy defined by Eq.~\eqref{PDP decomposition entropy}.

Our results complement those of Krumm \textit{et al.}
More precisely, we show that if there exists an entropy satisfying the condition in Eq.~\eqref{Krumm's condition for entropy}, then positive work cannot be extracted from any SPM cycle implementing decompositions into perfectly distinguishable states.
Moreover, combining Corollary \ref{cor:strong SPM condition} with the results of Krumm \textit{et al.}\ \cite{Krumm2015,Krumm2017}, it follows that positive work cannot be extracted from SPM cycles implementing the generalized L\"uders measurement processes in systems satisfying PDP decomposability and the strong pure transitivity, even without assuming that the input states are left undisturbed.

\subsubsection{Minagawa \textit{et al.}\ \cite{Minagawa2022}}

Minagawa et al. \cite{Minagawa2022} considered systems that are pure transitive and that admit at least one state possessing a PDP decomposition.
Following von Neumann's argument, they assumed that decompositions into PDP states can be implemented by an SPM without any thermodynamic cost.
In such systems, they showed that if there exists a state \(\rho^A\) admitting two or more distinct PDP decompositions \(\rho^{XA},\sigma^{YA}\in\mathrm{Ens}_{\mathrm{pure}}\p{\rho^A}\) that induce different entropies, \textit{i.e.},
\begin{equation}
    H^X\p{\mathrm{tr}_A\rho^{XA}}\ne H^Y\p{\mathrm{tr}_A\sigma^{YA}},
\end{equation}
then positive work can be extracted from an SPM cycle using that state.

However, they only provided an example of a system that admits a state with multiple PDP decompositions, and did not construct any example of a system that is both pure transitive and admits a state with multiple PDP decompositions.
We show that the regular hexagon model satisfies these conditions, and we provide the first explicit construction of an SPM cycle from which positive work can be extracted.
Furthermore, we construct another example of an SPM cycle from which positive work can be extracted by using the square model, which does not admits a state with multiple PDP decompositions.

\section{Discussion}

In this work, we formulated information thermodynamics in generalized probabilistic theories (GPTs) and showed that the second law of information thermodynamics holds in GPTs for entropies satisfying subadditivity.
We also showed that, in semipermeable-membrane (SPM) models, the extractable work is bounded by entropy changes associated with the measurement processes.
These results imply that, even in GPTs, work exceeding the restriction imposed by the second law of thermodynamics cannot be extracted as long as the processes involved are consistent with entropy nondecrease.

We also clarified situations in which such entropies naturally exist.
In particular, the measurement entropy satisfies the required properties for important classes of composite systems, including standard quantum tensor products and compositions in which all but one subsystem are classical.

Our results stand in contrast to previous studies on thermodynamics in GPTs.
Earlier works mainly assumed the existence of semipermeable membranes (SPMs) that implement certain measurements without thermodynamic cost and analyzed particular SPM cycles individually.
By contrast, we analyzed general SPM processes from the viewpoint of information thermodynamics and derived general conditions under which such cost-free measurements are consistent with thermodynamics.
In this sense, our work provides a unified framework for understanding previous SPM-based arguments in terms of entropy changes associated with measurement processes.

Furthermore, we provided the first explicit examples of GPT systems in which the existence of SPMs performing measurements for free can lead to violations of the second law of thermodynamics.
In our examples, positive work extraction originates either from the failure of fundamental entropy properties or from discrepancies between different definitions of entropy.
In particular, we derived conditions under which measurement processes do not decrease the measurement entropy and conditions under which forgetting the measurement outcomes does not decrease the accessible information entropy.
Therefore, if positive work can nevertheless be extracted under these conditions, the origin of the contradiction must lie in the mismatch between these two different definitions of entropy.
The square and regular hexagon systems provide concrete realizations of this mechanism.
However, these examples do not imply that the existence of such GPT systems themselves contradicts the second law of thermodynamics; rather, they rule out the existence of SPMs that implement the corresponding measurement processes without thermodynamic cost.

Our results also provide a new perspective on the role of the von Neumann entropy.
Von Neumann \cite{Neumann1927,Neumann1932} derived the entropy formula by assuming that projective measurement processes can be implemented without thermodynamic cost.
In contrast, our results suggest that the entropy itself may characterize the class of measurement processes that can be implemented thermodynamically for free.
In particular, the von Neumann entropy guarantees that L\"uders measurement processes are consistent with thermodynamics in quantum systems.
This further suggests the possibility that the von Neumann entropy may be characterized as the entropy that maximizes the class of ``thermodynamically free'' measurement processes.

On the other hand, several important open problems remain.
First, our sufficient conditions for the validity of the second law rely on structural assumptions such as subadditivity, and it remains unclear how generally entropies satisfying such conditions exist in GPTs.
In particular, when both the measured system and the memory system are nonclassical GPT systems, there may exist no nontrivial entropy for which the second law of information thermodynamics holds.
Second, although we analyzed SPM models from the viewpoint of information thermodynamics, the precise relation between the two frameworks is still not fully understood.
In particular, the results for SPM models are not direct consequences of the results for information thermodynamics, and a more unified understanding is desirable.
Finally, our results suggest that, in order to obtain genuinely constraining conditions on GPTs from thermodynamic consistency, it may be necessary to go beyond information-theoretic scenarios in which GPT systems are accessed only through measurements, and to investigate more general physical settings arising in thermodynamics and statistical mechanics.

\begin{acknowledgments}
The authors thank Hayato Arai and Francesco Buscemi for fruitful discussions.
K. O. was supported by Advanced Basic Science Course, a World-leading Innovative Graduate Study Program, the University of Tokyo.
S. U. was supported by Forefront Physics and Mathematics Program to Drive Transformation, a World-leading Innovative Graduate Study Program, the University of Tokyo.
H. T. was supported by JSPS Grants-in-Aid for Scientific Research No. JP25K00924, MEXT KAKENHI Grant-in-Aid for Transformative
Research Areas B ``Quantum Energy Innovation” Grant Numbers 24H00830 and 24H00831, JST FOREST No. JPMJFR2365, JST MOONSHOT No. JPMJMS256E and Royal Society International Collaboration Awards 2025 Flexigrant number ICA/R2/252240.
\end{acknowledgments}

\section*{Author contributions}

K. O. and S. U. contributed equally to this work.

\clearpage
\appendix
\widetext

\section{Proofs of the properties of entropies}
\label{app:entropies}
In this appendix, we provide the proofs of the properties of entropies in GPTs that are stated in the main text.

\begin{proof}[Proof of Proposition~\ref{prop:consistency with the direct sum}]
    For brevity, we write \(\overline{XA}\coloneq\bigoplus_{x\in X}A_x\).
    \begin{enumerate}
        \item
        Every effect \(e\in\mathrm{Eff}\p{\overline{XA}}\) can be written in the form \(e\p{\rho_x^{A_x}}_{x\in X}=\sum_{x\in X}e_x\rho_x^{A_x}\) where \(e_x\in\mathrm{Eff}\p{A_x}\).
        This implies that if \(e\) is indecomposable, then there exists some \(x\in X\) such that \(e\p{\rho_x^{A_x}}_{x\in X}=e_x\rho_x^{A_x}\).
        Thus, any fine-grained measurement \(\mathcal E\in\mathrm{Meas}_{\mathrm{fg}}\p{\overline{XA}}\) can be written as \(\mathcal E\p{\rho_x^{A_x}}_{x\in X}=\sum_{x\in X}\mathcal E_x\rho_x^{A_x}\) using fine-grained measurements \(\mathcal E_x\in\mathrm{Meas}_{\mathrm{fg}}\p{A_x}\) (\(x\in X\)).
        Therefore, we obtain
        \begin{equation*}
            S^{\overline{XA}}\p{\p{\rho_x^{A_x}}_{x\in X}}=H^X\p{\p{\V{\rho_x^{A_x}}}_{x\in X}}+\aab{S^{A_x}}_{x\in X}\p{\p{\rho_x^{A_x}}_{x\in X}}.
        \end{equation*}
        \item
        For an arbitrary classical system \(Z\), we have
        \begin{equation*}
            \begin{aligned}
                I_\mathrm{acc}^{Z:\overline{XA}}\p{\rho^{Z\overline{XA}}}&=\sup_{\mathcal E:\overline{XA}\to Y}I^{Z:Y}\p{\p{\mathrm{id}_{X'}\otimes\mathcal E}\rho^{Z\overline{XA}}}\\
                &=\sup_{\p{\mathcal E_x:A_x\to Y_x}_{x\in X}}I^{Z:\bigoplus_{x\in X}Y_x}\p{\p{\p{\mathrm{id}_Z\otimes\mathcal E_x}\rho_x^{ZA_x}}_{x\in X}}\\
                &=\sup_{\p{\mathcal E_x:A_x\to Y_x}_{x\in X}}\p{I^{Z:X}\p{\rho^{ZX}}+\aab{I^{Z:Y_x}}_{x\in X}\p{\p{\p{\mathrm{id}_Z\otimes\mathcal E_x}\rho_x^{ZA_x}}_{x\in X}}}\\
                &=I^{Z:X}\p{\rho^{ZX}}+\sum_{x\in X}\V{\rho_x^{ZA_x}}\sup_{\mathcal E_x:A_x\to Y_x}I^{Z:Y_x}\p{\p{\mathrm{id}_Z\otimes\mathcal E_x}\hat\rho_x^{ZA_x}}\\
                &=I^{Z:X}\p{\rho^{ZX}}+\aab{I_\mathrm{acc}^{Z:A_x}}_{x\in X}\p{\rho^{Z\overline{XA}}}.
            \end{aligned}
        \end{equation*}
        Therefore, we obtain
        \begin{equation*}
            \begin{aligned}
                {S'}^{\overline{XA}}\p{\rho^{\overline{XA}}}&=\sup_{\rho^{Z\overline{XA}}\in\mathrm{Ens}\p{\rho^{\overline{XA}}}}I_\mathrm{acc}^{Z:\overline{XA}}\p{\rho^{Z\overline{XA}}}\\
                &=\sup_{\rho^{Z\overline{XA}}\in\mathrm{Ens}\p{\rho^{\overline{XA}}}}\p{I^{Z:X}\p{\rho^{ZX}}+\aab{I_\mathrm{acc}^{Z:A_x}}_{x\in X}\p{\rho^{Z\overline{XA}}}}\\
                &=H^X\p{\rho^X}+\sum_{x\in X}\V{\rho_x^{A_x}}\sup_{\rho_x^{Z_xA_x}\in\mathrm{Ens}\p{\hat\rho_x^{A_x}}}I_\mathrm{acc}^{Z_x:A_x}\p{\rho_x^{Z_xA_x}}\\
                &=H^X\p{\rho^X}+\aab{{S'}^{A_x}}_{x\in X}\p{\rho^{\overline{XA}}}.
            \end{aligned}
        \end{equation*}
        \item
        If the entropy \(S\) is consistent with the direct sum, we find
        \begin{equation*}
            \begin{aligned}
                {S'}^{\overline{XA}}\p{\rho^{\overline{XA}}}&=\sup_{\rho^{Z\overline{XA}}\in\mathrm{Ens}\p{\rho^{\overline{XA}}}}\p{I_\mathrm{acc}^{Z:\overline{XA}}\p{\rho^{Z\overline{XA}}}+\aab{S^{\overline{XA}}}_Z\p{\rho^{Z\overline{XA}}}}\\
                &=\sup_{\rho^{Z\overline{XA}}\in\mathrm{Ens}\p{\rho^{\overline{XA}}}}\biggl(I^{Z:X}\p{\rho^{ZX}}+\aab{I_\mathrm{acc}^{Z:A_x}}_{x\in X}\p{\rho^{Z\overline{XA}}}\\
                &\hspace{10em}{}+\aab{H^X}_Z\p{\rho^{ZX}}+\aab{\aab{S^{A_x}}_{x\in X}}_Z\p{\rho^{Z\overline{XA}}}\biggr)\\
                &=H^X\p{\rho^X}+\sum_{x\in X}\V{\rho_x^{A_x}}\sup_{\rho_x^{Z_xA_x}\in\mathrm{Ens}\p{\hat\rho_x^{A_x}}}\p{I_\mathrm{acc}^{Z_x:A_x}\p{\rho_x^{Z_xA_x}}+\aab{S^{A_x}}_{Z_x}\p{\rho_x^{Z_xA_x}}}\\
                &=H^X\p{\rho^X}+\aab{{S'}^{A_x}}_{x\in X}\p{\rho^{\overline{XA}}}.
            \end{aligned}
        \end{equation*}
        \qedhere
    \end{enumerate}
\end{proof}

\begin{proof}[Proof of Proposition~\ref{prop:measurement entropy non-decresing measurement processes}]
We denote \(\mathcal M=\p{\mathcal M_x}_{x\in X}\) and take a fine-grained measurement \(\mathcal E^x=\p{e^x_y}_{y\in Y_x}\in\mathrm{Meas}_\mathrm{fg}\p A\) such that \(S_\mathrm{meas}\p{\hat{\mathcal M_x\rho^A}}=H^{Y_x}\p{\mathcal E^x\hat{\mathcal M_x\rho^A}}\) for each \(x\in X\).
Since \(\mathcal M\) is repeatable, we have
\begin{equation*}
    \p{e^{x'}_y\mathcal M_{x'}}\p{\mathcal M_x\rho^A}=\delta_{x,x'}e^x_y\p{\mathcal M_x\rho^A}.
\end{equation*}
Thus, the measurement \(\mathcal E\coloneq\p{e^{x'}_y\mathcal M_{x'}}_{x'\in X,y\in Y_{x'}}\) satisfies
\begin{equation*}
    H^{Y_x}\p{\mathcal E^x\hat{\mathcal M_x\rho^A}}=H^{\bigoplus_{x'\in X}Y_{x'}}\p{\mathcal E\hat{\mathcal M_x\rho^A}}.
\end{equation*}
Therefore, we obtain
\begin{equation*}
    \begin{aligned}
        S_\mathrm{meas}^{XA}\p{\mathcal M\rho^A}&=\tilde S_{\mathrm{meas}}^{XA}\p{\mathcal M\rho^A}\\
        &=H^X\p{\mathds1_A\mathcal M_x\rho^A}+\aab{S_{\mathrm{meas}}}_X\p{\mathcal M\rho}\\
        &=H^X\p{\mathds1_A\mathcal M_x\rho^A}+\aab{H^{\bigoplus_{x'\in X}Y_{x'}}}_X\p{\p{\mathrm{id}_X\otimes\mathcal E}\mathcal M\rho^A}\\
        &=H^{X\bigoplus_{x'\in X}Y_{x'}}\p{\p{\mathrm{id}_X\otimes\mathcal E}\mathcal M\rho^A}\\
        &=H^{\bigoplus_{x'\in X}Y_{x'}}\p{\mathcal E\rho^A}.
    \end{aligned}
\end{equation*}
Due to \(e^{x'}_y\mathcal M_{x'}\le\mathds1_A\mathcal M_{x'}\), \(e^{x'}_y\mathcal M_{x'}\) is indecomposable whenever \(\mathds1_A\mathcal M_{x'}\) is indecomposable.
This implies that \(\mathcal E\) is fine-grained.
Therefore, we obtain \(S_\mathrm{meas}^A\p{\rho^A}\le S_\mathrm{meas}^{XA}\p{\mathcal M\rho^A}\).
\end{proof}

\begin{proof}[Proof of Proposition~\ref{prop:entropy nondecreasing under information-preserving processes}]
Suppose that \(\rho^{XA}\in\mathrm{Ens}\p{\rho^A}\) and a process \(\mathcal G:B\to A\) satisfies \(\p{\mathrm{id}_X\otimes\mathcal G\mathcal F}\rho^{XA}=\rho^{XA}\).
Then,
\begin{equation*}
    I^{X:A}_\mathrm{acc}\p{\rho^{XA}}\ge I^{X:B}_{\mathrm{acc}}\p{\p{\mathrm{id}_X\otimes\mathcal F}\rho^{XA}}
    \ge I^{X:A}_{\mathrm{acc}}\p{\p{\mathrm{id}_X\otimes\mathcal G}\p{\mathrm{id}_X\otimes\mathcal F}\rho^{XA}}
    =I^{X:A}_{\mathrm{acc}}\p{\rho^{XA}},
\end{equation*}
and hence \(I^{X:B}_{\mathrm{acc}}\p{\p{\mathrm{id}_X\otimes\mathcal F}\rho^{XA}}=I^{X:A}_\mathrm{acc}\p{\rho^{XA}}\).
Therefore, since \(S\) is nondecreasing under information-preserving processes,
\begin{equation*}
    {S'}^B\p{\mathcal F\rho^A}\ge I^{X:B}_\mathrm{acc}\p{\p{\mathrm{id}_X\otimes\mathcal F}\rho^{XA}}+\aab{S^B}_X\p{\p{\mathrm{id}_X\otimes\mathcal F}\rho^{XA}}\ge I^{X:A}_{\mathrm{acc}}\p{\rho^{XA}}+\aab{S^A}_X\p{\rho^{XA}}.
\end{equation*}
for every \(\rho^{XA}\in\mathrm{Ens}\p{\rho^A}\). This implies \({S'}^B\p{\mathcal F\rho^A}\ge{S'}^A\p{\rho^A}\).
\end{proof}

\begin{proof}[Proof of Proposition~\ref{prop:strogly repeatable implies information preserving}]
We first show that, for any state \(\rho^A\in\mathrm{St}\p{A}\), a strongly repeatable measurement \(\mathcal{M}=\p{\mathcal{M}_x}_{x\in X}\) satisfies \(\p{\mathrm{id}_Y\otimes \mathcal{M}\,\mathrm{tr}_X}\sigma^{YXA}=\sigma^{YXA}\) for any \(\sigma^{YXA}\in \mathrm{Ens}\p{\mathcal{M}\rho}\).
Since \(\sum_{y\in Y}\sigma^A_{x,y}=\mathcal{M}_x\rho^A\), we obtain
\begin{equation*}
    \p{\mathrm{id}_Y\otimes\mathcal M}\p{\mathrm{tr}_X\sigma^{YXA}}
    =\p{\mathcal M_x\sum_{x'\in X}\sigma^A_{x',y}}_{x\in X,y\in Y}
    =\p{\sigma^A_{x,y}}_{x\in X, y\in Y}
    =\sigma^{YXA}.
\end{equation*}

We next assume that, for any state \(\rho^A\in\mathrm{St}\p{A}\), a measurement \(\mathcal{M}=\p{\mathcal{M}_x}_{x\in X}\) satisfies \(\p{\mathrm{id}_Y\otimes \mathcal{M}\,\mathrm{tr}_X}\sigma^{YXA}=\sigma^{YXA}\) for any \(\sigma^{YXA}\in \mathrm{Ens}\p{\mathcal{M}\rho}\), and show that \(\mathcal{M}\) is strongly repeatable.
For each unnormazied state \(0\le\sigma^A\le\mathcal{M}_x\rho^A\), we define an ensemble of states \(\sigma^{YXA}=\p{\sigma^A_{x',y}}_{x'\in X,y\in Y}\in \mathrm{Ens}\p{\mathcal{M}\rho^A}\) by
\begin{equation*}
    Y\coloneq\B{0,1},\quad\sigma_{x',y}^A\coloneq\left\{
    \begin{aligned}
        &\sigma^A&\p{x'=x,y=0}\\
        &\mathcal{M}_x\rho^A-\sigma^A&\p{x'=x,y=1}\\
        &0 &\p{x'\ne x,y=0}\\
        &\mathcal{M}_{x'}\rho^A&\p{x'\ne x,y=1}
    \end{aligned}
    \right..
\end{equation*}
Then, from \(\p{\mathrm{id}_Y\otimes \mathcal{M}\,\mathrm{tr}_X}\sigma^{YXA}=\sigma^{YXA}\), we obtain
\begin{equation*}
    \mathcal{M}_x\sigma^A=\mathcal{M}_x\mathrm{tr}_X\p{\sigma_{x',0}^A}_{x'\in X}=\sigma_{x,0}^A=\sigma^A.
\end{equation*}
\end{proof}

\section{Proof of the generalized Sagawa--Ueda inequalities}
\label{app:Sagawa--Ueda}

In this appendix, we provide the proof of the generalized Sagawa--Ueda inequalities presented in the main text.
\begin{proof}[Proof of Theorem~\ref{thm:Sagawa--Ueda}]
For brevity, we write \(F_{\beta\p2}^A\coloneq F_\beta^A\p{\rho_{\p2}^A}\) and similarly for other quantities.
First, we have
\begin{equation*}
    \begin{aligned}
        W_\mathrm{ext}&=-F_{\beta\p2}^A-F_{\beta\p2}^{B_2}+F_{\beta\p0}^A+F_{\beta\p0}^{B_2}-\beta^{-1}\p{S^A_{\p2}+S^{B_2}_{\p2}-S^A_{\p0}-S^{B_2}_{\p0}}\\
        &\le-F_{\beta\p2}^A+F_{\beta\p0}^A-\beta^{-1}\p{S^A_{\p2}+S^{B_2}_{\p2}-S^A_{\p0}-S^{B_2}_{\p0}}\\
        &\le-F_{\beta\p2}^A+F_{\beta\p0}^A-\beta^{-1}\p{\tilde S^{KAB_2}_{\p2}-H^K_{\p2}-S^A_{\p0}-S^{B_2}_{\p0}}\\
        &=-\Delta F_\beta^A-\beta^{-1}\p{\Delta S_\mathcal{F}-I}.
    \end{aligned}
\end{equation*}
Next,
\begin{equation*}
    \begin{aligned}
        W_\mathrm{meas}&=F_{\beta\p1}^{KM}+F_{\beta\p1}^{B_1}-F_{\beta\p0}^M-F^{B_1}_{\p0}+\beta^{-1}\p{\tilde S^{KM}_{\p1}+S^{B_1}_{\p1}-S^M_{\p0}-S^{B_1}_{\p0}}\\
        &\ge F_{\beta\p1}^{KM}-F_{\beta\p0}^M+\beta^{-1}\p{\tilde S^{KM}_{\p1}+S^{B_1}_{\p1}-S^M_{\p0}-S^{B_1}_{\p0}}\\
        &\ge F_{\beta\p1}^{KM}-F_{\beta\p0}^M+\beta^{-1}\p{\tilde S^{KM}_{\p1}+\tilde S^{KAMB_1}_{\p1}-\tilde S^{KAM}_{\p1}-S^{AMB_1}_{\p0}+S^A_{\p0}}\\
        &\ge F_{\beta\p1}^{KM}-F_{\beta\p0}^M+\beta^{-1}\p{H^K_{\p1}+\tilde S^{KAMB_1}_{\p1}-\tilde S^{KA}_{\p1}-S^{AMB_1}_{\p0}+S^A_{\p0}}\\
        &=\aab{\Delta F_\beta^M}-\beta^{-1}\p{H-I-\Delta S_\mathcal{M}},
    \end{aligned}
\end{equation*}
where \(F_\beta^{KM}\coloneq E^M-\beta^{-1}\tilde S^{KM}\).
Finaly,
\begin{equation*}
    \begin{aligned}
        W_\mathrm{eras}&=F_{\beta\p3}^M+F_{\beta\p3}^{B_3}-F_{\beta\p2}^{KM}-F_{\beta\p2}^{B_3}+\beta^{-1}\p{S^M_{\p3}+S^{B_3}_{\p3}-\tilde S^{KM}_{\p2}-S^{B_3}_{\p2}}\\
        &\ge F_{\beta\p3}^M-F_{\beta\p2}^{KM}+\beta^{-1}\p{S^M_{\p3}+S^{B_3}_{\p3}-\tilde S^{KM}_{\p2}-S^{B_3}_{\p2}}\\
        &\ge F_{\beta\p3}^M-F_{\beta\p2}^{KM}+\beta^{-1}\p{S^{MB_3}_{\p3}-\tilde S^{KMB_3}_{\p2}}\\
        &=-\aab{\Delta F_\beta^M}+\beta^{-1}\p{H+\Delta S_\mathcal{V}}.
    \end{aligned}
\end{equation*}
\end{proof}

\section{Proofs for the SPM cycle}
\label{app:SPM model}
In this appendix, we provide proofs of the statements on the SPM cycle presented in the main text.

\begin{proof}[Proof of Lemma \ref{lem:SPM extractable work}]
    From Eq.~\eqref{eq:SPM extractable work}, we have
    \begin{equation*}
        \begin{aligned}
            W_\mathrm{ext}^{\p i}
            &=N\beta^{-1}\Biggl(\eval{\aab{H^{Z'^{\p i}_y K^{\p i}_y}}_{y\in Y^{\p i}}}_{\p{i,2}}-\eval{\aab{H^{Z^{\p i}_y K^{\p i}_y}}_{y\in Y^{\p i}}}_{\p{i,1}}\\
            &\hspace{6em}{}-\sum_{x\in X^{\p i}}\V{\rho_{\p{i,0}x}}\log V_x+\sum_{x\in X^{\p{i+1}}}\V{\rho_{\p{i+1,0}x}}\log V_x\Biggr),
        \end{aligned}
    \end{equation*}
    where \(W_\mathrm{ext}^{\p i}\) denote the maximal extractable work in the \(i\)-th step.
    Then we obtain
    \begin{equation*}
        \begin{aligned}
            W_\mathrm{ext}&\coloneq\sum_{i=0}^{n-1}W_\mathrm{ext}^{\p i}=N\beta^{-1}\sum_{i=0}^{n-1}\p{\eval{\aab{H^{Z'^{\p i}_y K^{\p i}_y}}_{y\in Y^{\p i}}}_{\p{i,2}}-\eval{\aab{H^{Z^{\p i}_y K^{\p i}_y}}_{y\in Y^{\p i}}}_{\p{i,1}}}\\
            &\le N\beta^{-1}\sum_{i=0}^{n-1}\p{\eval{\aab{H^{Z'^{\p i}_y K^{\p i}_y}+\aab{S^A}_{Z'^{\p i}_yK^{\p i}_y}}_{y\in Y^{\p i}}}_{\p{i,2}}-\eval{\aab{H^{Z^{\p i}_y K^{\p i}_y}+\aab{S^A}_{Z^{\p i}_yK^{\p i}_y}}_{y\in Y^{\p i}}}_{\p{i,1}}}\\
            &=N\beta^{-1}\sum_{i=0}^{n-1}\p{\eval{\tilde S^{\bigoplus_{y\in Y^{\p i}}Z'^{\p i}_yK^{\p i}_yA}}_{\p{i,2}}-\eval{\tilde S^{\bigoplus_{y\in Y^{\p i}}Z^{\p i}_yK^{\p i}_yA}}_{\p{i,1}}}\\
            &=N\beta^{-1}\sum_{i=0}^{n-1}\p{\eval{\tilde S^{\bigoplus_{y\in Y^{\p i}}Z^{\p i}_yA}}_{\p{i,0}}-\eval{\tilde S^{\bigoplus_{y\in Y^{\p i}}Z^{\p i}_yK^{\p i}_yA}}_{\p{i,1}}+\eval{\p{\tilde S^{\bigoplus_{y\in Y^{\p i}}Z'^{\p i}_yK^{\p i}_yA}-\tilde S^{\bigoplus_{y\in Y^{\p i}}Z'^{\p i}_yA}}}_{\p{i,2}}}\\
            &=N\beta^{-1}\sum_{i=0}^{n-1}\p{\eval{\aab{\aab{S^A}_{Z^{\p i}_y}}_{y\in Y^{\p i}}}_{\p{i,0}}-\eval{\aab{\aab{\tilde S^{K^{\p i}_yA}}_{Z^{\p i}_y}}_{y\in Y^{\p i}}}_{\p{i,1}}+\eval{\aab{\aab{\tilde S^{K^{\p i}_yA}-S^A}_{Z'^{\p i}_y}}_{y\in Y^{\p i}}}_{\p{i,2}}}.
        \end{aligned}
    \end{equation*}
\end{proof}

\begin{proof}[Proof of Theorem \ref{thm:weak SPM condition}]
    By definition, we have \(\rho_{\p{i,1}yz'}^{K^{\p i}_yA}=\mathcal M^{\p i}_y\rho_{\p{i,0}yz'}^A\).
    Moreover, since \(\mathcal M^{\p i}_y\) is repeatable, \(\rho_{\p{i,2}yz'}^{K^{\p i}_yA}=\mathcal M^{\p i}_y\rho_{\p{i,2}yz'}^A\) also holds.
    Therefore, it follows from Lemma \ref{lem:SPM extractable work} that Eq.~\eqref{eq:weak SPM condition} implies \(W_{\mathrm{ext}}\le0\).
\end{proof}

\bibliography{ref}

\end{document}